\documentclass[referee,pdflatex,sn-mathphys-ay]{sn-jnl}

\usepackage{graphicx}%
\usepackage{multirow}%
\usepackage{amsmath,amssymb,amsfonts,bbm}%
\usepackage{amsthm}%
\usepackage{mathrsfs}%
\usepackage[title]{appendix}%
\usepackage{xcolor}%
\usepackage{textcomp}%
\usepackage{manyfoot}%
\usepackage{booktabs}%
\usepackage{algorithm}%
\usepackage{algorithmicx}%
\usepackage{algpseudocode}%
\usepackage{mathtools}
\usepackage[capitalize,nameinlink]{cleveref}
\usepackage{subfig,bm,todonotes}

\allowdisplaybreaks

\newcommand{\R}{\mathbb{R}}
\renewcommand{\P}{\mathbb{P}}
\newcommand{\E}{\mathbb{E}}

\theoremstyle{thmstyleone}%
\newtheorem{theorem}{Theorem}
\newtheorem{lemma}[theorem]{Lemma}

\newtheorem{counterexample}[theorem]{Counterexample}

\theoremstyle{definition}
\newtheorem{defn}[theorem]{Definition}

\usepackage{comment}

\graphicspath{{Figures/}}

\begin{document}

%%%%%%%%%%%%%%%%%%%%%%%%%%%%%%%%%%%%%%%%%%%%%%%%%%%%%%%%%%%%%%%%%%%%%%%%%%%%%%

\title[When lookout sees crackle]{When lookout sees crackle: Anomaly~detection~via~kernel~density~estimation}

\author[1]{\fnm{Rob J} \sur{Hyndman}}
\author*[2]{\fnm{Sevvandi} \sur{Kandanaarachchi}}\email{sevvandi.kandanaarachchi@csiro.au}
\author[3]{\fnm{Katharine} \sur{Turner}}
\affil[1]{\orgname{Monash University}, \country{Australia}}
\affil[2]{\orgname{CSIRO}, \country{Australia}}
\affil[3]{\orgname{Australian National University}, \country{Australia}}

\abstract{
  We present an updated version of \textit{lookout} -- an algorithm for detecting anomalies using kernel density estimates with bandwidth based on Rips death diameters -- with theoretical guarantees. The kernel density estimator for updated \textit{lookout} is shown to be consistent, and the proposed multivariate scaling is robust and efficient. We show our updated algorithm performs better than the previous version on diverse examples.
  %TBC. 200 or fewer words.
}


\maketitle

% -----------------------------------------------------------
\section{Introduction} \label{sec:intro}
% ------------------------------------------------------------

Anomaly detection is used in many application domains such as credit card fraud and scam detection \citep{kolupuri2025scams}.  \cite{lookout} proposed the \textit{lookout} algorithm for anomaly detection using a kernel density estimate with bandwidth based on persistence homology. Anomalies are identified by fitting a Generalized Pareto Distribution (GPD) to the negative logarithm of the smallest leave-one-out kernel density estimates, and declaring points with low GPD probabilities as anomalies.

In this paper, we explore some of the theoretical properties of the lookout algorithm, particularly the choice of bandwidth based on the Rips death diameters, and the consistency of the kernel density estimates produced by this choice. This leads us to propose a modified lookout algorithm using a different bandwidth, although still based on the Rips death diameters. The modified algorithm improves the performance and robustness of the original algorithm, under a broader range of conditions.

First we introduce some background material, before describing both the original lookout algorithm, and our proposed modifications.

\subsection{Anomalies and surprisals}

There is no broadly accepted definition for an anomaly. \citet{hawkins1980identification} defines an anomaly or outlier as \textit{an observation which deviates so much from other observations as to arouse suspicion it was generated by a different mechanism.} \citet{Barnett1978} define an anomaly or outlier as \textit{an observation (or a subset of observations) which appears to be inconsistent with the remainder of that set of data.} Both these definitions agree that anomalies are rare and different from other observations. Usually that is applied by fitting some (conditional) probability distribution to the set of observations, and considering whether each observation can be considered a likely random draw from the probability distribution. We formalize this in the following definition.

\begin{defn}[Anomalies]
  Given a set of observations $\bm{y}_{i} \in \R^m$ for $i \in \{1, \dots, n\}$, and a probability density function $f$, we define observation $\bm{y}_i$ as anomalous if
  \begin{equation}
    \text{Pr}(f(Y) \le f(\bm{y}_{i})) < \alpha,
  \end{equation}
  where $Y$ is a random variable with density $f$, and $\alpha>0$ is a chosen threshold.
\end{defn}

\begin{defn}[Surprisals]
  Given a probability density function $f$, the \emph{surprisal} of an observation $\bm{y}$ drawn from the density $f$ is defined as $-\log f(\bm{y})$.
\end{defn}
These are better known as  ``log scores'' in statistics \citep{Good1952-zw}, but we prefer the term ``surprisals'' from information theory \citep{Stone2022}, as it provides a measure of how surprising an observation is, given the probability density function $f$. Specifically, it takes large values for observations that are unlikely under $f$.

Let $S = -\log f(\bm{Y})$ be the random variable representing the surprisal of a random vector $\bm{Y}$ drawn from $f$, and let $G(s) = \text{Pr}(S\le s)$ be its distribution function. Then the surprisal probability of observation $\bm{y}_i$ is given by
$$p_i = 1 - G(s_i^-),$$
where $s_i = -\log f (\bm{y}_i)$, and $G(s_i^-) = \lim_{t \uparrow s_i} G(t)$ is the left limit of $G$ at $s_i$. An observation is an anomaly if and only if $p_i < \alpha$.

\subsection{Extreme Value Theory}

Consider $n$ independent and identically distributed random variables $X_1, \dots, X_n$ having a distribution function $G$, with their maximum $M_n = \max \{X_1, \dots, X_n\} $. Extreme Value Theory (EVT) analyses and characterizes the behavior of $M_n$. One of the fundamental results of EVT is the Fisher-Tippet-Gnedenko Theorem. We restate a version of it from \citet[][Theorem 3.1.1]{coles2001introduction}.

% \begin{equation}\label{eq:evt1}
%     M_n = \max \{X_1, \, \ldots, \, X_n\} \, .
% \end{equation}
% If $F$ is known, the distribution of $M_n$ can be derived for any value of $n$ as
% \begin{align}\label{eq:evt2}
%     P\{M_n \leq z \} & = P\{X_1 \leq z, \, \ldots, \, X_n \leq z \} \, ,  \\
%     & = P\{X_1 \leq z\} \times \ldots \times P\{X_n \leq z\}\, , \\
%     & =  \left(F(z)\right)^n \, .
% \end{align}
% However, $F$ is not known in practice. This gap is filled by Extreme Value Theory, which studies approximate families of models for $F^n$ so that extremes can be modeled and uncertainty quantified. Similar to the Central Limit Theorem for means, the Fisher-Tippet-Gnedenko Theorem, also known as the the extreme value theorem or the extremal types theorem states that under certain conditions, a scaled maximum $\frac{M_n - a_n}{b_n}$, if $a_n$ and $b_n$ exist, have certain limit distributions.

\begin{theorem}[Fisher-Tippett-Gnedenko] \label{thm:FisherTippett}
  If there exist sequences of constants $\{a_n>0\}$ and $\{b_n\}$ such that
  $$ P\left\{ (M_n - b_n)/a_n \leq z \right\} \rightarrow H(z) \quad \text{as} \quad n \to \infty,
  $$
  for a non-degenerate distribution function $H$, then $H$ is a member of the Generalized Extreme Value (GEV) family
  \begin{equation}\label{eq:EVT4}
    H(z) = \exp\left\{ -\left[ 1 + \xi\left(\frac{z - \mu}{\sigma} \right)\right]^{-1/\xi} \right\}\, ,
  \end{equation}
  where the domain of the function is $\{z: 1 + \xi (z - \mu)/\sigma >0 \}$ and the parameters $\mu, \xi \in \R$ and $\sigma > 0$.
\end{theorem}
The parameters $\mu$, $\sigma$ and $\xi$ are called the location, scale and shape parameters respectively. The shape parameter $\xi$ determines the behavior of the tails and differentiates between the Gumbel, Fréchet and Weibull distributions: a Fréchet distribution is obtained when $\xi>0$, a Weibull distribution is obtained when $\xi < 0$, and a Gumbel distribution is obtained in the limit as $\xi\rightarrow 0$. Importantly, when $F$ has a finite upper bound, then $\xi < 0$.

A second result from EVT that we will use is Pickands theorem \citep{Pickands1975}, here restated from \citep[][Theorem 4.1]{coles2001introduction}.

\begin{theorem}[Pickands] Let $X_1, \dots, X_n$  be a sequence of independent random variables with a common distribution function $G$, and let
  $M_n = \max \{X_1, \ldots, X_n \}$. Suppose $G$ satisfies Theorem~\ref{thm:FisherTippett}, so that for large $n$, $P\{ M_n \leq z \} \approx H(z)$,
  where
  $$
  H(z) = \exp\left\{ -\left[ 1 + \xi\left(\frac{z - \mu}{\sigma} \right)\right]^{-1/\xi} \right\}\, ,
  $$
  for some $\mu, \xi \in \R$ and $\sigma >0$. Then for large enough $u$,
  \begin{equation}\label{eq:POT4}
    \text{Pr}(X \le u+x \mid X > u) = 1 - \left( 1 + \frac{\xi x}{\sigma_u} \right)^{-1/\xi}\, ,
  \end{equation}
  where $x > 0$, $(1 + \xi x)/\sigma_u >0$, and $\sigma_u = \sigma + \xi(u- \mu)$.
\end{theorem}
The family of distributions defined by equation~\eqref{eq:POT4} is called the \textbf{Generalized Pareto Distribution} (GPD). The shape parameter $\xi$ is the same in both GPD and GEV distributions. Furthermore, the GEV parameters can be recovered from the GPD formulation by estimating a GPD on all observations above some threshold.

\subsection{Kernel Density Estimation}

Given a set of observations $\bm{y}_{i} \in \R^m$ for $i \in \{1, \dots, n\}$, a kernel density estimate is defined as \citep{chacon2018multivariate}
\begin{equation}\label{eq:mkde}
  \hat{f_n}(\bm{y}) = \frac{1}{n} \sum_{i=1}^n |\bm{H}|^{-1/2} K\big(\bm{H}^{-1/2}(\bm{y} - \bm{y}_i)\big),
\end{equation}
where $K$ is a square-integrable spherically-symmetric function, bounded below by 0, with a finite second-order moment and unit integral, and $\bm{H}$ is a symmetric $m\times m$ positive-definite matrix.

The expected $L_2$ distance between the estimator $\hat{f}_n$ and the true density $f$ is given by the mean integrated squared error (MISE) defined as
\begin{equation}\label{eq:mise}
  \text{MISE}(\hat{f}_n) = \text{E}\left[\int \left(\hat{f_n}(\bm{u}) - f(\bm{u})\right)^2 d\bm{u}\right].
\end{equation}
We say that $\hat{f}$ is a \emph{consistent estimator} of $f$ if $\text{MISE}(\hat{f}_n)\rightarrow 0$ as $n\rightarrow\infty$.

We will need some conditions on the density function, the kernel and the bandwidth matrix to ensure consistency of the kernel density estimator. These are given below.

\begin{description}
  \item[A1] The density function $f$ is square integrable and twice differentiable, with all of its second-order partial derivatives bounded, continuous and square integrable.
  \item[A2] The kernel $K$ is a square-integrable, spherically symmetric function on $\R^m$ that satisfies\break $\int K(\bm{u})d\bm{u} = 1$, $\int \bm{u} K(\bm{u})d\bm{u} = \bm{0}$ and $\int \bm{u}\bm{u}^T K(\bm{u})d\bm{u} = \mu_2(K)\bm{I}_m$ for some $\mu_2(K) > 0$.
  \item[A3] The bandwidth matrix $\bm{H}$ is a sequence of symmetric positive-definite matrices such that\break $\bm{H}\rightarrow \bm{O}$ and $n|\bm{H}|^{1/2} \rightarrow \infty$ as $n\rightarrow\infty$.
\end{description}

The following result is provided by \citet[][p31--37]{chacon2018multivariate}.

\begin{theorem}[Consistency] Let $\bm{y}_1,\dots,\bm{y}_n$ be a sequence of independent random variables from a distribution on $\R^m$ that satisfies A1. Then under conditions A2 and A3,
  $$
  \lim_{n\rightarrow\infty} \text{E}\left[\int \left(\hat{f_n}(\bm{u}) - f(\bm{u})\right)^2 d\bm{u}\right] = 0.
  $$
\end{theorem}

We consider matrices of the form $\bm{H} = h_n \bm{I}_m$, where $h_n\in\R$ is a function of the sample $\{\bm{y}_1,\dots,\bm{y}_n\}$. We have $|\bm{H}|=h_n^m$ so $n|\bm{H}|^{1/2} \rightarrow \infty$ precisely when $n(h_n)^{m/2} \to \infty$. This leads to our definition of an admissible bandwidth as these are the bandwidths which will determine consistent estimators.
\begin{defn}[Admissible bandwidth]\label{def:admissible_bandwidth}
  We say that $h_n$ is an \emph{admissible bandwidth} for samples in $\R^m$ if $h_n\rightarrow 0$ and $n(h_n)^{m/2}\rightarrow \infty$ as $n\rightarrow\infty$ almost surely.
\end{defn}

\subsection{The lookout algorithm for anomaly detection}

The \textit{lookout} algorithm of \cite{lookout} has four inputs: the dataset $\{\bm{y}_1,\dots,\bm{y}_n\}$, where each $\bm{y}_i \in \R^m$; a parameter $\beta \in (0,1)$ indicating the proportion of points to be used for fitting the Generalized Pareto Distribution (GPD); a parameter $\alpha \in (0,1)$ signifying the cut-off threshold for determining an anomaly; and a logical variable \textit{unitize} indicating whether to scale the data. We briefly outline the steps below.

\begin{algorithm}
  \caption{Lookout algorithm}\label{alg:lookout}
  \begin{algorithmic}[1]
    \Require Data $\{\bm{y}_1, \dots, \bm{y}_n\}$, proportion $\beta \in (0,1)$, threshold $\alpha \in (0,1)$, \emph{unitize} $\in \{\text{True}, \text{False}\}$.
    \State If \textit{unitize}, scale the data to lie in the unit cube in $\R^m$.
    \State Compute the persistence homology barcode of the (possibly scaled) data cloud for dimension zero using the Vietoris-Rips diameter \citep{ghrist2008barcodes}.
    \State From the barcode, obtain the ordered finite death diameters $\{d_i\}_{i = 1}^{n-1}$ for connected components, and their successive differences $\Delta d_i = d_{i+1} - d_i$.
    \State Let $d_* = d_{i_*}$ where $i_* = \text{arg\! max} \{\Delta d_i \}_{i=1}^{n-2}$; that is, $d_*$ is the diameter corresponding to the lower point of the maximum $\Delta d_i$.
    \State Compute the kernel density estimates at the observations: $f_i = \hat{f}(\bm{y}_i)$, $i=1,\dots,n$, where $\hat{f}$ is given by \eqref{eq:mkde}, $\bm{H} = h_n\bm{I}_m$, and $h_n = d_*$.
    \State Compute the leave-one-out kde values $f_{-i} = \frac{1}{n-1} \big( n f_i - \bm{H}^{-1/2} K(\bm{0}) \big)$ for $i=1,\dots,n$.
    \State Fit a Generalized Pareto Distribution (GPD) to the largest $1-\beta$ of the surprisals $\{-\log f_i\}_{i=1}^n$, obtaining estimates of the location, scale and shape parameters $(\mu, \sigma, \xi)$.
    \State Calculate the probability of each observation by applying the estimated GDP to the leave-one-out kde values: $p_i = (1-\beta)P(-\log f_{-i} \mid \hat{\mu}, \hat{\sigma}, \hat{\xi})$.
    \State If $p_i < \alpha$, declare $\bm{y}_i$ anomalous.
  \end{algorithmic}
\end{algorithm}

The lookout algorithm provides density-based anomaly detection, where anomalies are identified as points with low density in the space of the data. If there are no anomalies in the data, the algorithm should identify $\alpha$ of the data as false positive ``anomalies''.

While \citet{lookout} demonstrated excellent results in applying their algorithm to a variety of simulated and real datasets, we see several potential areas for improvement that we aim to address in this paper. First, elements of the algorithm are not robust to outliers, including the min-max scaling, and the use of the largest gap in Rips death diameters. Second, the diagonal bandwidth matrix $\bm{H}$ does not account for correlations between variables, which can lead to suboptimal performance if the data have strong correlations. Third, it is not obvious that this bandwidth matrix choice will lead to a consistent kernel density estimator, particularly for distributions with heavy tails. Finally, the GPD estimation made minimal assumptions about the distribution of the surprisal values, but we know that these are bounded above and below, and that fact can be used to improve the estimated GPD.

In the following section, we will explore the conditions under which the kernel density estimator is consistent for $h_n=d_{\omega(n)}$ (the $\omega(n)$ smallest death radius), and the implications for the lookout algorithm. In particular, we indicate it is unlikely the previous choice of $h_n=d_*$ will be consistent for distributions with heavy tails such as a power distribution, although it seems likely that they are satisfied for lighter-tailed distributions such as the Gaussian distribution. We will also show that an alternative choice of bandwidth, based on an upper quantile of the Rips death diameters ($d_{\gamma(n-1)}$), satisfies these conditions for a wide range of distributions. We provide lower and upper bounds on $d_{\gamma(n-1)}$ (with high probability) showing it is approximately of order $n^{1/m}$.

An important mathematical tool for proving when a bandwidth is admissible uses the notion of crackle, introduced in \cite{Crackle}, which describes the persistent homology of noise. They proved asymptotic results for the persistent homology of the union of unit balls centered on iid samples of various distributions. There are qualitative differences between sets of samples taken from the Gaussian distribution (which do not ``crackle'') compared to those from heavy-tailed distributions. We will adapt the methods to this setting where we wish to shrink the radius to $0$, albeit slowly.

%This paper proved asymptotic results about the

\subsection{The modified lookout algorithm}

As a result of these theoretical considerations, we propose three modifications to the lookout algorithm, described in detail in \cref{sec:modified}.

\begin{enumerate}
  \item Rather than use a max-min scaling of the data, we standardize the data using a robust covariance estimator. That is, we scale and rotate the data to ensure that the covariance matrix of the transformed data is approximately the identity matrix. This ensures that the data have zero pairwise correlations, which leads to more efficient kernel density estimates. It also leads to each variable having approximately unit variance, which provides a more robust scaling than min-max scaling.
  \item Rather than using the lower point of the largest gap in Rips death diameters, we use an upper quantile of the Rips death diameters. This avoids the problem where anomalies affect the largest gap in the Rips diameters that cause a persistent homology death. It also ensures the consistency of the kernel density estimates under relatively weak conditions.
  \item The estimated GPD shape parameter is constrained to be non-positive, which improves the stability of the GPD fitting. This is justified by noting that the kernel density estimates are bounded, so the limiting distribution under the Fisher-Tippett-Gnedenko theorem is a Weibull distribution, which has a negative shape parameter, $\xi$.
\end{enumerate}

Note that we do not need to estimate the density of the underlying distribution directly. Instead, we can estimate the density of a transformation of the data, and identify anomalies in the transformed space. Hence, there is no need to undo the scaling or rotation.

\subsection{Outline of the paper}

The remainder of this paper is organized as follows. In \cref{sec:theory}, we explore the theoretical properties of the lookout algorithm, particularly the choice of bandwidth based on Rips death diameters, and the consistency of the kernel density estimates produced by this choice. In \cref{sec:modified}, we propose a modified lookout algorithm that improves performance and robustness under a broader range of conditions. We provide some empirical experiments, demonstrating the performance of the modified lookout algorithm on simulated and real datasets, first in contrast to the original lookout (\cref{sec:comparewitholderversion}), and then in contrast to other anomaly detection methods (\cref{sec:compareWithOtherMethods}). Finally, in \cref{sec:conclusion}, we summarize our results and discuss future work.

\section{Lower and upper bounds on $d_{\omega_n}$}\label{sec:theory}

We will first prove a number of limit theorems about $d_{\omega_n}$, that is the $\omega_n$-th death time in the degree 0 Rips persistent homology of the Rips filtration (when ordered increasingly by size) in order to construct an admissible bandwidth $h_n=d_{\omega_n}$. %Recall that For consistency of the kernel density estimator, we want to show that for bandwidth $h_n$ that $h_n \to 0$ and $nh_n^m\to \infty$. This means we need both (asymptoptic) lower and upper bounds on $d_{\omega_n}$.

\subsection{Limit of $d_{n-1}$ for Gaussian Samples}

We will first consider the case where the samples come from a Gaussian distribution. Since $d_{\omega_n}\leq d_{n-1}$, to asymptotically upper bound $d_{\omega_n}$ it is sufficient to show $d_{n-1} \to 0$ as $n\to \infty$. The proof is a straightforward modification of the proof for Theorem 2.1 in \cite{Crackle}. The only real change is that we need to consider a sequence of decreasing radii for the balls, whereas in \cite{Crackle} the radii remain constant (value $1$). We just need to ensure that the radii decrease slowly enough so that the union of balls remains connected. Due to this variation in the proof, we include the calculations for completeness.
%For consistency in the case of
% Gaussian distributions over $\R^2$, we also need to show that the
% index function for the maximum gap in death values will go to
% infinity almost surely.

\begin{theorem}\label{thm:d_n-1 Gaussian}
  Let $S_n$ denote a random sample of $n$ iid points from a fixed multivariate Gaussian in $\R^m$. Then for the Rips filtration over $S_n$ we have $d_{n-1} \to 0$ as $n\to \infty$ almost surely.
\end{theorem}
\begin{proof}
  We will first prove the result in the case of a standard Gaussian, where the probability density function is spherically symmetric and is $c \exp\{-\|x\|^2/2\}$ for some constant $c$. Denote $f(r) = c\exp\{-r^2/2\}$ for this density in terms of the radial coordinate.

  %  Set $r(n)$ to satisfy $r(n)^d c \log \log n = 1$. By construction $r(n)\rightarrow 0$ as $n \rightarrow \infty$. We will use $r(n)$ to show that the core is covered by balls around the samples with radius $r(n)$.
  Set $R_n:=\sqrt{2(\log n+m\log \log n)}$ and  $C_n=\sqrt{2(\log n -\log \log n)}$.
  $R_n$ will be the radius outside which we will expect no samples to occur and $C_n$ will represent the radius of the core of the point cloud - which we expect to be densely covered.

  Cover $B(0, C_n)$ by a grid with side lengths $g_n$. Let $Y(n, C_n, g_n)$ be the event that at least one of the boxes in this grid is empty. Equation 3.1 in \cite{Crackle} states:
  $$
  \mathbb{P}(Y(n,C_n,g_n)) \leq (2g_n^{-1})m (C_n)^m\exp\{-ng_n^m f(C_n)\}.
  $$
  %Set $R_n^c=\sqrt{2(\log n -\log \log n)}$, which will represent the core of the point cloud.
  Set $g_n$
  %=\frac{r(n)}{2\sqrt{d}}$
  so that $g_n^m=\frac{1}{c \log \log n}$.
  In addition, note that $f(C_n)= c \exp\{-\log n+\log\log n\}=c\frac{\log n}{n}$.
  Combined, we get
  \begin{align*}
    \mathbb{P}(Y(n,C_n,g_n))
    & \leq 2^{m}c \log \log n (2(\log n -\log \log n))^{m/2}\exp\left(\frac{-nc\log n}{cn\log \log n}\right)
    \leq k\frac{(\log n)^{m/2}\log\log n }{\exp(\frac{\log n}{\log \log n})}
  \end{align*}
  where $k$ is a constant independent of $n$. For large $n$, $\log n>(m\log\log n) \log \log n$ and thus
  $$
  \exp\left(\frac{\log n}{\log \log n}\right)>\exp(m\log\log n)=(\log n)^m.
  $$
  Thus, for large $n$ we can bound
  $$
  \mathbb{P}\left(Y\left(n,C_n, g_n\right)\right)\leq k \frac{\log\log n }{(\log n)^{m/2}}
  $$
  which goes to $0$ as $n \to \infty$.

  %Set $R_n^2:=\sqrt{2(\log n+d\log \log n)}$.

  Let $X(n, R_n)$ be the event that all the samples lie inside $B(0, R_n)$.
  From Theorem 4 in \cite{Crackle} we have $\lim_{n\to \infty} \E[|S_n \cap B(0,R_n)^c|]=0$. As $|S_n \cap B(0,R_n)^c|$ must be an integer between $0$ and $n$ this implies that $\lim_{n\to \infty} \mathbb{P}(X(n,R_n))=1$.

  %that the probability of there being a point outside the ball centered at the origin with radius $R_n:=\sqrt{2(\log n +d\log\log n)}$ goes to zero as $n\to \infty$. This implies that we can assume all points are within $B(0,R_n)$ and, from our grid size, that every point in $B(0, C_n)$ is within $\frac{2\sqrt{d}}{ c \log \log n}$ of a sample point.

  We want to bound the distance from any point in $B(0, R_n)$ to its closest point in $B(0, C_n)$ which is the same as bounding $R_n-C_n$.
  We have $R_n+C_n= \sqrt{2(\log n +m\log\log n)}+\sqrt{2(\log n -\log \log n)} > 2\sqrt{\log n}$ so
  \begin{align*}
    R_n-C_n &=\frac{R_n^2-C_n^2}{R_n+C_n}\\
    &<\frac{2(\log n +m\log\log n)-2(\log n -\log \log n)}{2\sqrt{\log n}}\\
    &=\frac{(m+1)\log\log n}{\sqrt{\log n}}
  \end{align*}

  Suppose that both $Y(n, C_n, g_n)$ does not hold (every grid box in the covering the core contains at least one sample) and $X(n, R_n)$ does hold (all samples within a ball of radius $R_n$). That is every sample lies in $B(0,R_n)$ and there is a sample in every box in the grid over $C_n$. Set
  $$
  r(n):= \frac{2\sqrt{m}}{(c\log \log n)^{1/m}} + \frac{(m+1)\log\log n}{\sqrt{\log n}}
  $$
  which was chosen so $r(n)>(R_n-C_n) + 2\sqrt{m}g_n$.
  By construction we have $\cup_{x\in S_n}B(x, r(n))$ must be connected. This is because every sample is within $(R_n-C_n)$ of a grid point and every grid point is within $2\sqrt{m}g_n$ of a sample point.

  This implies that $d_{n-1}\leq r(n)$ as $d_{n-1}$ is the supremum of radii $r$ where $\cup_{x\in S_n}B(x, r)$ is disconnected.

  %This then implies that every point within $R_n^c$
  %Set $r(n)$ to satisfy $r(n)^d c \log \log n = 1$. By construction $r(n)\rightarrow 0$ as $n \rightarrow \infty$.

  %  Let $s(n)=r(n) + (d+1)\frac{\log\log n}{\sqrt{\log n}}$.
  %  Let $X(n,r(n))$ be the event that $\cup_{x\in S_n} B(x, s(n))$ is connected.
  As both $Y(n, C_n, g_n)$ and $X(n, R_n)$ occur with probability $1$ as $n$ grows to infinity, we have $d_{n-1}\leq r(n) \to 0$ as $n\to \infty$ almost surely.

  %  We want to show that %there is a function $r(n)$ such that $r(n)\rightarrow 0$ as $n \rightarrow \infty$ and
  %  $\mathbb{P}(X(n,s(n)))\rightarrow 1$ as $n\rightarrow \infty$.
  %  This implies that $d_{n-1}\leq s(n)$ must also limit to $0$.

  %  %Let $s(n)=r(n) + (d+1)\frac{\log\log n}{\sqrt{\log n}}$.
  %  Then by the above arguments $\P(X(n,s(n)))\rightarrow 1$ as $n\rightarrow \infty$. This implies that $0<d_{n-1}<s(n)$. We also have $s(n)\to 0$ as $n\to \infty$. Together, they squeeze $d_{n-1}$ to satisfy $d_{n-1}\to 0$ as $n\to \infty$.

  To then extend the result to general multivariate Gaussian distributions, observe that our sample of points $\hat{S}_n$ is the linear transformation $T$ of a sample of points $S_n$ of the standard Gaussian. If $\cup_{x\in S_n} B(x, r)$ is connected, then we can infer that $\cup_{x\in S_n} B(T(x), \|T\|r)$ is also connected. This implies that $d_{n-1}(\hat{S}_n)\leq \|T\|d_{n-1}(S_n)$ and thus $d_{n-1}(\hat{S}_n) \to 0$ as $n\to \infty$.
\end{proof}

\subsection{A counterexample where $d_{\omega_n}$ does not converge to $0$}
%Consequences of crackle for distributions with fat tails}

There are many potential indexing functions $\omega_n$ that we can choose to decide which death time to use to compute the bandwidth matrix. We are interested in using a bandwidth that is a function $f(d_{\omega_n})$ for some strictly increasing nice function $f$ (with $f(0)=0$). Thus, in order for this choice $f(d_{\omega_n})$ to be an admissible bandwidth we need $d_{\omega_n} \to 0$ as $n \to \infty$.

If our samples are Gaussian, then we know that $d_{\omega_n}$ satisfies this limit. However, when samples come from a fat-tailed distribution, we can construct examples where with positive probability $d_{\omega_n}$ does not converge to $0$. The methods here are highly influenced by those within \cite{Crackle} but with a different focus we we are more concerned in bounding the probability that there are at least a specified number of connected components instead of limit theorems of expected number of components. 

%Consider the distribution on $\R^m$ with the density function $f(x)=\frac{c_m}{1+\|x\|^{m+1}}$ where $c_m$ is the appropriate constant. 

For the remainder of this subsection we will consider the distribution on $\R^m$, $m\geq 3$, with the density function $f(x)=\frac{c_m}{1+\|x\|^{m+1}}$ where $c_m$ is the appropriate constant. %Set $R_n=\frac{2n^{2/m}}{c_m s_{m-1}}$ where $s_{m-1}$ is the volume of the $(m-1)$-sphere.

A useful quantity is the number of samples with radius at least $R$. Let $A_n^{[R, \infty)}$ be the number of samples lying in $\{x\in \mathbb{R}^m\mid \|x\|\geq R\}.$

We will want to show that outside a certain radius $R_n$ that all the points are isolated in the Rips complex at length scale $1$ and thus each contributing their own connected component. A sufficient condition for this to hold is for there to be no edges in the Rips complex at length scale $1$ between points of radius at least $R_n-1$. Note that we will be not claiming anything about the points in the annular region from radius $R_n-1$ to $R_n$ being isolated as they could be connected to points with smaller radius.

\begin{lemma}\label{lem:bound edges}
Fix a sequence of positive numbers $R_n$.
Let $S_n$ be a set of $n$ point sampled i.i.d. according to density $f$. Let $g(n)$ be any function bounded above by $\sqrt{\frac{R_n^{m-1}}{\log n}}$.
Let $Y_n$ be the number of edges in the Rips complex at radius $1$ between points in $\{x\in \mathbb{R}^m\mid \|x\|>R_n-1\}$.
We have the limit $$\P\left(Y_n=0 \mid A^{[R_n-1, \infty)}_n < \sqrt{\frac{R_n^{m-1}}{\log n}} \text{ and } A^{[R_n, \infty)}>g(n)\right)\to 1$$ as $n\to \infty$.
\end{lemma}
\begin{proof}

%Consider the Rips complex at radius $1$. We want to lower bound the probability that all the points outside $B(0,R)$ are isolated. A sufficient condition for this is that there are no edges between points outside $B(0, R-1)$.

\begin{align*}
  \mathbb{P}(\|x-y\|\leq 1|\mid \|x\|, \|y\|\geq R-1)&=
  \frac{1}{\mathbb{P}(\|x\|,\|y\|\geq R-1)}\int_{R-1}^\infty \frac{c_m s_{m-1}r^{m-1}}{1+r^{m+1}} \mathbb{P}(y\in B(x,1)\mid |x|=r) \,dr
\end{align*}

If $|x|=r$ ($r$ sufficiently large), then the probability that another sample $y$ lies in $B(x,1)$ is bounded above by $\int_{B(x,1)} \frac{c_m}{1+\|z\|^{m+1}}dz \leq C/r^m$ for some constant $C$ independent of $r$. This implies that there exists a constant $C_1$ independent of $R$ such that
\begin{align*}
  \mathbb{P}(\|x-y\|\leq 1|\mid \|x\|, \|y\|>R-1)& < (R-1)^2 C_1\int_{R-1}^\infty \frac{1}{r^{m+2}}\,dr.
\end{align*}
Thus there is some constant $C_2$ independent of $R$ (for sufficiently large $R$) such that
$$
\mathbb{P}(\|x-y\|\leq 1|\mid \|x\|, \|y\|>R-1)\leq \frac{C_2}{R^{m-1}}.
$$

Let $Y_n$ be the number of edges in the Rips complex at radius $1$ between points outside $B(0,R_n-1)$. %When $S_n\leq 2n^{1-2/m}$ and $\hat{S}_n\leq n^{1-2/m}$ then the number of samples outside $B(0, R-1)$ is upper bounded by $3n^{1-2/m}$. 
For each pair of points the probability of an edge between them is bounded above by $C_2/R_n^{m-1}$. Changing the function $g(n)$ will not invalidate this upper bound. Together we get the upper bound on expectation
\begin{align*}
  \E\left[Y_n|A^{[R_n-1, \infty)}_n\leq \sqrt{\frac{R_n^{m-1}}{\log n}} \text{ and } A^{[R_n, \infty)}>g(n)\right]
  \leq \left(\frac{R_n^{m-1}}{\log n}\right)\frac{C_2}{2R_n^{m-1}}
  = \frac{C_2}{2\log n}
\end{align*}
for some $C_2$ independent of (sufficiently large) $n$.

Since $Y_n$ can only be a non-negative integer we have $$\mathbb{P}\left(Y_n\neq 0|A^{[R-1, \infty)}_n\leq \sqrt{\frac{R_n^{m-1}}{\log n}} \text{ and } A^{[R_n, \infty)}>g(n)\right)\leq \mathbb{E}\left[Y_n|A^{[R_n-1, \infty)}_n\leq \sqrt{\frac{R_n^{m-1}}{\log n}}\right].$$ As $\E\left[Y_n|A^{[R_n-1, \infty)}_n\leq \sqrt{\frac{R_n^{m-1}}{\log n}} \text{ and } A^{[R_n, \infty)}>g(n)\right]\to 0$ as $n\to \infty$ we can conclude that $$\mathbb{P}(Y_n = 0|A^{[R_n-1, \infty)}_n\leq \sqrt{\frac{R_n^{m-1}}{\log n}}\text{ and } A^{[R_n, \infty)}>g(n))\to 1$$  as $n\to \infty$.

%Fix any $\delta_2\in (0,1)$. We have $\E[Y_n|S_n<2n^{1-2/m}  \text{ and } \hat{S}_n\leq n^{1-2/m}]<1-\delta_2$ for sufficiently large $n$. As $Y_n$ is a non-negative integer we get $\mathbb{P}(Y_n=0\mid S_n<2n^{1-2/m} \text{ and } \hat{S}_n\leq n^{1-2/m})>\delta_2$ for sufficiently large $n$.

\end{proof}

Given the above lemma, our next step is to show that with positive probability that we have enough points with radius at least $R_n$ but not too many of radius more than $R_n-1$. 

\begin{lemma}\label{lem:bound points}
Let $S_n$ be a sample of points i.i.d. according to $f$.
Set $R_n=\frac{n^{2/m}}{c_m s_{m-1}}$ where $s_{m-1}$ is the volume of the $(m-1)$-sphere.

Then $\mathbb{P}\left( A^{[R_n-1, \infty)}_n < 3n^{1-2/m} \text{ and } A^{[R_n, \infty)}_n > \frac{1}{4}n^{1-2/m}\right)>1/3$ for $n$ sufficiently large.
\end{lemma}

\begin{proof}
For each radius $R$ the probability a sample point has modulus at least $R$ (that is, it lies outside $B(0, R)$) is
$$
p_{[R,\infty)}=\int_{R}^\infty \frac{c_m s_{m-1}r^{m-1}}{1+r^{m+1}}\,dr.
$$

 We have $p_{[R, \infty)}\in (\frac{c_m s_{m-1}}{2R}, \frac{c_m s_{m-1}}{R})$ for $R>1$. Thus $A^{[R_n, \infty)}_n \sim \mathcal{B}(n, p)$, the binomial distribution for some $p \in (\frac{n^{-2/m}}{2}, n^{-2/m})$. Using Chebyshev's inequality
$$\mathbb{P}\left(A_n^{[R_n, \infty)}\in \left(\frac{n^{1-2/m}}{4},2n^{1-2/m}\right)\right)>\frac{1}{2}$$.

Given that a sample point $x$ lies in $B(0,R)$, the probability that $x$ has its modulus in $[R-1, R)$ (that is, $x\in B(0, R)\backslash B(0, R-1)$) is
\begin{align*}
  \hat{p}&=\frac{\int_{R-1}^R \frac{c_m s_{m-1}r^{m-1}}{1+r^{m+1}}\,dr }{1-p_{[R, \infty)}}
  % &<\frac{\int_{R-1}^R \frac{c_m s_{m-1}}{r^2}\,dr }{1-p_{R}}\\
  <\frac{-\frac{c_m s_{m-1}}{R}+\frac{c_m s_{m-1}}{R-1}}{1-\frac{c_m s_{m-1}}{R}}
  <\frac{2c_m s_{m-1}}{R^2}
\end{align*}
for sufficiently large $R$. For $R_n=\frac{2n^{2/m}}{c_m s_{m-1}}$ we have $\hat{p}_n\leq C_1 n^{-4/m}$ for large $n$ (with $C_1$ a constant independent of $n$).

Let $\hat{A}^{[R_{n-1}, R_n)}_n$ the number of samples in $\{x\in \mathbb{R}\mid \|x\|\in[R_{n-1}, R_n)\}$ conditioned on there being $A^{[R_n, \infty)}_n$ samples outside $B(0,R_n)$. We have $\hat{A}^{[R_{n-1}, R_n)}_n \sim \mathcal{B}(n-A^{[R_n, \infty)}_n, \hat{p})$. We can couple $\mathcal{B}(n-A^{[R_n, \infty)}_n, \hat{p}_n)$ to $\mathcal{B}(n, \hat{p}_n)$ to imply that $\mathbb{P}(\hat{A}^{[R_n, \infty)}_n\leq n^{1-2/m})>2/3$ for large enough $n$. 

Combining
\begin{align*}
    \mathbb{P}&\left( A^{[R_n-1, \infty)}_n < 3n^{1-2/m} \text{ and } A^{[R_n, \infty)}_n > \frac{1}{4}n^{1-2/m}\right) \\
    &\geq  \mathbb{P}\left( A^{[R_n-1, \infty)}_n < 3n^{1-2/m} \text{ and } A^{[R_n, \infty)}_n \in \left( \frac{1}{4}n^{1-2/m}, 2n^{1-2/m}\right)\right)\\
    & \geq \mathbb{P}\left( \hat{A}^{[R_n-1, \infty)}_n < n^{1-2/m}\mid  A^{[R_n, \infty)}_n \in \left( \frac{1}{4}n^{1-2/m}, 2n^{1-2/m}\right)\right)\mathbb{P}\left( A^{[R_n, \infty)}_n \in \left( \frac{1}{4}n^{1-2/m}, 2n^{1-2/m}\right)\right)\\
    &\geq \frac{2}{3}\frac{1}{2}=\frac{1}{3}
\end{align*}
\end{proof}

We can finally combine these to complete our counterexample. Note that the probabilities are far from tight as we only need to show that $d_{\omega_n}$ does not converge to $0$

\begin{counterexample}
For $\omega_n=n-\frac{1}{4}n^{1-2/m}$ we have
$\mathbb{P}(d_{\omega_n}>1)>1/4$ and thus $d_{\omega_n}$ is not an admissible bandwidth.
\end{counterexample}
\begin{proof}
A sufficient condition for $d_{omega_n}\geq 1$ is for the Rips complex at length scale $1$ to have $\frac{1}{4}n^{1-2/m}$ isolated points. 

One way to orchestrate this many isolated points is to have $A_n^{[R_n, \infty)}>\frac{1}{4}n^{1-2/m}$ and $Y_n=0$. Consider $R_n=\frac{n^{2/m}}{c_m s_{m-1}}$ where $s_{m-1}$ is the volume of the $(m-1)$-sphere.

From \ref{lem:bound points} we have $\mathbb{P}\left( A^{[R_n-1, \infty)}_n < 3n^{1-2/m} \text{ and } A^{[R_n, \infty)}_n > \frac{1}{4}n^{1-2/m}\right)>1/3$ for $n$ sufficiently large. 

By substituting in $R_n$ we can easily check that  $3n^{1-2/m}<\sqrt{\frac{R_n^{m-1}}{\log n}} $. This means we can apply Lemma \ref{lem:bound edges} to say
that for $Y_n$ the number of edges between samples of radius $R_n-1$, we have $\mathbb{P}\left(Y_n=0\mid A^{[R_n-1, \infty)}_n < 3n^{1-2/m} \text{ and } A^{[R_n, \infty)}_n > \frac{1}{4}n^{1-2/m}\right)to 1$ as  $n\to \infty$, and thus at least $\frac{3}{4}$ for sufficiently large $n$.

Combining everything, 
\begin{align*}
    \mathbb{P}(d_{\omega_n}>1)&\geq \mathbb{P}\left(Y_n=0 \text{ and } A^{[R_n-1, \infty)}_n < 3n^{1-2/m} \text{ and } A^{[R_n, \infty)}_n > \frac{1}{4}n^{1-2/m}\right)\\
    \geq \frac{1}{3}\frac{3}{4}=\frac{1}{4}.
\end{align*}

This contradicts $d_{\omega_n}\to 0$ as $n\to \infty$ with probability $1$ and thus $d_{\omega_n}$ can't be an admissible bandwidth. 
\end{proof}

\subsection{Upper bounding the quantile death times $d_{\gamma(n-1)}$}\label{sec:upperboundquartile}

Despite the counterexample in the previous subsection, we can find sequences of death times that will limit to zero under very mild assumptions of the distribution (even with fat tails!).

An approach that generally works is to use a quantile of the death times. That is $d_{\gamma(n-1)}$ for some $\gamma\in [0,1]$. Recall that there are only $n-1$ finite death times, so $d_{n-1}$ is the largest finite death value. Strictly speaking, we need an integer value, so you can interpret $\gamma(n-1)$ as either the floor or the ceiling. More precisely in this subsection we will consider sequences $\omega_n$ where there exists $\gamma\in (0,1)$ with $\omega_n\leq \gamma n$ for large enough $n$.

%The following is a standard useful lemma, but we include it for completeness.

%\begin{lemma}\label{lem:samplescover}
%  Let $U_1, \ldots, U_l\subset \R^m$ and $X$ a random sample such that $\P(X\in U_i)>\delta $ for all $i$. Let $\{x_1, \ldots x_n\}$ be i.i.d. draws of $X$. If $Z_i$ is the event that at no samples $x_j$ lie in $U_i$ then
% $\P(\cup Z_i) \leq l e^{-n\delta}.$
%\end{lemma}
%\begin{proof}
%  For a single draw $x_j$ we have $\P(x_j\notin Z_i)\leq 1-\delta\leq e^{-\delta}$ and thus for $n$ draws $\P(Z_i)\leq e^{-n\delta}$. Combining over all the $Z_i$ we gain the upper bound
%  $\P(\cup Z_i)\leq \sum_i \P(Z_i)\leq le^{-n\delta}$.
%\end{proof}

It is useful to have notation for offsets of Euclidean subsets. For any subset $L\subset \R^m$ and any $\delta>0$, let $L^\delta$ denote the set $\{y+v \mid y\in L, \|v\| \leq \delta\}$.

\begin{lemma}\label{lem:covernumber}
  Let $f: \R^m\to [0,\infty)$ be a Lipschitz probability density function (with Lipschitz constant $M$), and $L=f^{-1} [c,\infty)$ be a non-empty superlevel set of $f$ for some $c>0$. For all $\varepsilon\in(0, c/(2M)]$ there exists $\{x_1, \ldots x_l\}\subset L$ such that $L \subseteq \cup_i B(x_i, \varepsilon)$, $l\leq\frac{2^{m}}{c b_m \varepsilon^m}$.
  %and $\P(X\in B(x_i, \varepsilon))>\frac{c b_m}{2} \varepsilon^m$ for all $i$ (where $X$ is drawn according to pdf $f$).
  Here $b_m$ is the volume of the unit ball in $\R^m$.
  %Furthermore, $\cup_i B(x_i, \varepsilon)$ contains at most $$??$$
  % connected components.
\end{lemma}
\begin{proof}
  Fix $\varepsilon \in (0, c/M]$ and take a maximal $(\varepsilon/2)$-separated subset $S = \{x_1, x_2, \ldots, x_l\}$ within $L$. This means that $|x_i-x_j|\geq \varepsilon$ for all $i\neq j$ and that no additional point in $L$ can be added to $S$ without breaking this condition. By construction this implies that $L\subseteq \cup_{i} B(x_i, \varepsilon)$ as otherwise any missed location in $L$ could be added as an additional point and the set remain $(\varepsilon/2)$-separated.

  As each $x_i\in L$ we also have $B(x_i, \varepsilon/2)\subset L^{\varepsilon/2}$. Since $f$ is Lipschitz with Lipschitz constant $M$, we have
  $L^{\varepsilon/2} \subseteq f^{-1}[c-M\varepsilon/2, \infty)$.
  For $\varepsilon<c/M$ we have $B(x_i, \varepsilon/2)\subset f^{-1}[c/2, \infty)$ so $$\P(X\in B(x_i, \varepsilon/2))=\int_{B(x_i, \varepsilon/2)}f dx \geq \frac{c b_m \varepsilon^m}{2^{m}}.$$ Each of the $B(x_i, \varepsilon/2)$ are disjoint so $l\frac{c b_m \varepsilon^m}{2^{m}} \leq \P(X\in \cup_i B(x_i, \varepsilon/2))\leq 1$. We thus get $l\leq \frac{2^{m}} {c b_m \varepsilon^m}$.

  % As $B(x_i, \varepsilon) \subset f^{-1}[c-M\varepsilon, \infty)\subset f^{-1}[c/2, \infty]$ we have $\P(X\in B(x_i, \varepsilon))\geq \frac{c b_m \varepsilon^m}{2}$.
\end{proof}

\begin{theorem}\label{thm:d_alphan to 0}

  Let $d_k$ denote the $k$th smallest death time of the $0$-homology for the Rips filtrations of a set of $n$ points sampled from probability distribution $X$ over $\R^m$ such that the probability density function $f$ exists and for over the region where $f\leq 1$ it is Lipschitz with Lipschitz constant $M$ . Note that there are $n-1$ finite death times. Fix a $\gamma\in (0,1)$ and choose any sequence of integers $\{\omega_n\}$ such that $0\leq \omega_n \leq \gamma n$ for all sufficiently large $n$. Then there exists a constant $C(f,\gamma)$ (dependent on $f$ and $\gamma$ but not $n$) such that
  $$
  \mathbb{P}\left(d_{\omega_n} \leq C(f,\gamma)n^{-1/m}\right) < 1- \frac{12}{n(1-\gamma)}
  $$
  which limits to $1$ as $n\to \infty$.
  %Then for any sequence $\{g_n\}$, such that $g_n\to 0$ as $n \to \infty$, we have
  %  $$\lim_{n\to \infty}\mathbb{P}\left(d_{\omega_n} \geq n^{-1/m}g_n\right)=1.$$
\end{theorem}

%  Let $f:\R^m \to \R$ be a Lipschitz probability density function (with Lipschitz constant $M$) for $X$. For any $\gamma\in (0,1)$, the limit $d_{\gamma (n-1)} \to 0$ converges in probability, where $d_{\gamma (n-1)}$ is the $\gamma$ quantile of finite death times.

%   $$\lim_{n\to \infty}\mathbb{P}\left(d_{\gamma(n-1)} \leq  \left( \frac{\log n}{n}\right)^{1/m}\right)=1.$$ %almost surely, with $C(f,\gamma)$ a constant dependent on $f$ and $\gamma$.
%In particular, $$d_{\gamma(n-1)}<C(f,\gamma) \left( \frac{\log n}{n}\right)^{1/m}$$ almost surely, with $C(f,\gamma)$ a constant dependent on $f$ and $\gamma$.
%\end{theorem}

\begin{proof}
  %  Fix $\gamma\in [0,1)$.
  %Note that $d_{\omega_n}<d_{\gamma n}$ for sufficiently large $n$ so it is sufficient to prove the theorem
  Choose $c\in(0,1)$ such that for $L=f^{-1}[c, \infty)$ we have $\P(X\notin L)<(1-\gamma)/3$.

  Consider the sequence $r_n:= \left(\frac{2^{m+1}}{c b_m n(1-\gamma)}\right)^{1/m}$
  %so that $\frac{2^m}{c b_m r_n^m} =n(1-\gamma)/2.$
  (where $b_m$ is the volume of the unit ball in $\R^m$).
  Note that for $r_n \in (0, c/M)$ for large enough $n$. Restrict the sequence to $n$ where this holds.

  By \cref{lem:covernumber} there exists a cover of $L$ by $l_n$ balls $\{B(y_i, r_n)\mid i=1, 2, \ldots l_n\}$, with
  $$
  l_n\leq \frac{2^m}{c b_m r_n^m} =n(1-\gamma)/2 .
  %\frac{2^{m}}{c b_m r_n^m}=\frac{2^m n}{\log n}.
  $$
  %and
  %$\P(X\in B(y_i, r_n))>\frac{c b_m r_n^m}{2} = \frac{\log n}{n}$ for each $i$.
  %We have $d_{\gamma(n-1)}<r_n$ whenever $\cup_i B(x_i, r_n)$ has less than $n(1-\gamma)$ connected components.
  %  If $Z_i$ is the event that there are no samples in $B(y_i, r_n)$. By \cref{lem:samplescover}
  % $$
  % \P(\cup Z_i) \leq \frac{2^m n}{\log n} e^{-n \frac{\log n}{n}}= \frac{2^m}{\log n} .$$
  %This implies that as $n\to \infty$ we have at least one sample in each of the $B(y_i, r_n)$ with probability $1$.

  All points with the same ball $B(y_i, r_n)$ will be connected in the Rips complex with radius $2r_n$. This implies that the number of connected components when restricting to points inside $L$ will be bounded above by $l_n \leq n(1-\gamma)/2 $.

  %Let $p=\P(X\notin L)$. Note that for all $x\in L$ we have
  % $f(x)>c$. Since $f$ has only finitely many local maxima we have
  % that the number of connected components in $L$ is finite.
  %Denote these connected components by $L_1, L_2, \ldots L_K$.
  % %Consider also the sets $L_i^\varepsilon=\{y\in \R^m| \text{ there
  % exists }x\in L_i \text{ such that } \|x-y\|<\varepsilon\}$
  %We also know that the volumes of the $L_i$ are all bounded.

  Let $Y(L,n)$ denote the random variable for the number of our $n$ sample points that are not in $L$.
  In the case we would have an upper bound on the number of connected components in $\cup_j B(x_j, r_n)$ by $
  Y(L,n)+n(1-\gamma)/2$.
  We have $d_{\omega_n}<r_n$ whenever $\cup_i B(x_i, \epsilon_n)$ has less than $n-\omega_n$ connected components.
  We will want to show that
  $$
  Y(L,n)+n(1-\gamma)/2<n-\omega_n
  $$
  for large $n$. As $n-\omega_n>n(1-\gamma)$, it is sufficient to show $Y(L,n)<n(1-\gamma)/2$ for large $n$.

  From the construction of $L$ we have $Y(L,n)$ is a binomial distribution with mean $np$ and variance $np(1-p)$, with $p<(1-\gamma)/3$. We can use Chebyshev's inequality to bound $\P(Y(L,n)>n(1-\gamma)/2)$.
  \begin{align*}
    \P(Y(L,n)>n(1-\gamma)/2) & \leq
    \P(|Y(L,n)-np|>n(1-\gamma)/2-n(1-\gamma)/) \\
    & = \P(|Y(L,n)-np| >n(1-\gamma)/6)                    \\
    & \leq \frac{Var(Y)}{(n(1-\gamma)/6)^2}             \\
    & \leq \frac{36np(1-p)}{n^2(1-\gamma)^2}             \\
    %&\leq \frac{9(1-\gamma)(1-p)/2}{n(1-\gamma)^2}\\
    & \leq \frac{12}{n(1-\gamma)}
  \end{align*}
  This clearly goes to zero as $n$ goes to infinity.

  If there are less than $n(1-\gamma)$ connected components in $\cup_j B(x_j, r_n)$ then $d_{\omega_n}\leq r_n=: C(f, \gamma) n^{-1/m}.$
  Thus we have proved $\P(d_{\omega_n} > C(f, \gamma) n^{-1/m}
  )\leq \frac{12}{n(1-\gamma)} \to 0$ as $n\to \infty$
  %as $n\to \infty$.
  %As $r_n \to 0$ as $n\to \infty$ we have $d_{\gamma(n-1)} \to 0$ almost surely.
  %In particular, $$d_{\omega_n}\leq C(f,\gamma) \left( \frac{\log n}{n}\right)^{1/m}$$ almost surely,
  with $C(f,\gamma)$ a constant dependent on $f$ and $\gamma$ (and independent of $n$).
\end{proof}

\subsection{Lower bounding the quantile death times $d_{\gamma(n-1)}$}

Given the general upper bound result in Subsection \ref{sec:upperboundquartile} we will now focus showing limit theorems when our sequence of death times is a quantile $\gamma(n-1)$ for some $\gamma\in (0,1)$. Again we will pose our theorems in terms of $d_{\omega_n}$ where $\omega_n$ is controlled appropriately --- here $\omega_n\geq \gamma n$ for large enough $n$.

%Main mathematical results with proofs (possibly relegated to the Appendices).

%A consequence of the below theorems is that $nd_{\gamma(n-1)}\to
% \infty$ %whenever $d_*$ is chosen as the any fixed quantile of
% the death times.
%This holds for very general distributions. We only require a global
% upper bound on the pdf and the dimension of the space to be at least $2$.

\begin{theorem}\label{thm:nd_f(n)}
  Let $d_k$ denote the $k$th smallest death time of the $0$-homology for the Rips filtrations of a set of $n$ points sampled from probability distribution $X$ over $\R^m$ such that the probability density function $f$ exists and is bounded above. Note that there are $n-1$ finite death times. Fix a $\gamma\in (0,1)$ and choose any sequence of integers $\{\omega_n\}$ such that $\gamma n \leq \omega_n <n$ for all sufficiently large $n$. Then for any sequence $\{g_n\}$, such that $g_n\to 0$ as $n \to \infty$, we have
  $$\lim_{n\to \infty}\mathbb{P}\left(d_{\omega_n} \geq n^{-1/m}g_n\right)=1.$$
\end{theorem}

\begin{proof}
  %It is sufficient for there to exist a function $r(n)$ (with and $r(n)\to \infty$ as $n\to \infty$) such that $\mathbb{P}(d_{\omega_n}\geq \frac{r(n)}{n}) \to 1$ as $n\to \infty$.
  For any sequence of positive numbers, $\{r_n\}$, let $C(n, r_n)$ denote the number of components of $\cup_x B(x, r_n)$ (where $B(x,r_n)$ denotes the ball of radius $r_n$ centered at $x$). Each death time corresponds to the reduction by one of the number of connected components. This implies that $d_{\omega_n} > r_n$ occurs when $C(n, r_n) > n - \omega_n$.

  Now $C(n, r_n)$ is a random variable with range contained in $[1,n]$. This implies we can the upper bound the expectation by
  \begin{align*}
  \E[C(n, r_n)] &\leq n\left(1-\mathbb{P}\left(C(n, r_n )\leq n-\omega_n)\right)\right)+ (n-\omega_n)(\mathbb{P}(C(n,r_n) \leq (n- \omega_n))\\
    &= n - \omega_n \mathbb{P}(C(n,r_n)\leq n-\omega_n).
  \end{align*}
  Rearranging we get
  $$
  \mathbb{P}(C(n,r_n)\leq  n-\omega_n)\leq \frac{n-\E[C(n,r_n)]}{\omega_n}= \frac{\E[n-C(n,r_n)]}{\omega_n}.
  $$
  % \begin{align*}
  %     \E[X] &\leq N\mathbb{P}(X\geq (N-\lambda)) + (N-\lambda)(1-\mathbb{P}(X\geq (N- \lambda)))= N-\lambda + \lambda \mathbb{P}(X\geq N-\lambda)
  %\end{align*}
  %and hence $$\mathbb{P}(X\geq N-\lambda)\geq \frac{\E[X]-N +\lambda}{\lambda}.$$
  % Since the number of components can only be an integer between $1$ and $n$,
  %\begin{align*}
  %   \Mathbb{E}\left[C\left(n, \frac{r(n)}{n}\right)> n-\omega_n\right]&\leq n\mathbb{P}\left(C\left(n, \frac{r(n)}{n}\right)> n-\omega_n\right) + (1-\mathbb{P}\left(C\left(n, \frac{r(n)}{n}\right)> n-\omega_n\right) (n-\omega_n)\\
  %     &=(n-\omega_n)+ \mathbb{P}\left(C\left(n, \frac{r(n)}{n}\right)> n-\omega_n\right)\omega_n
  %\end{align*}
  %and thus
  %  $$
  %  \mathbb{P}\left(C\left(n, \frac{r(n)}{n}\right)\leq n-\omega_n\right) \leq \frac{\mathbb{E}[C(n, \frac{r(n)}{n})]-n}{\omega_n}.
  %  $$
  % It is thus sufficient to show that there exists a suitable sequence $r(n)$ such that $r(n)\to \infty$ as $n\to \infty$, and
  % $$
  % \lim_{n\to \infty}\frac{\mathbb{E}[C(n, \frac{r(n)}{n})]-n}{\omega_n} = 0.
  % $$
  A useful upper bound for $n-C(n,r_n)$ is $E(n,r_n)$, where $E(n,r_n)$ is the number of edges in the graph connecting any two points are less than $2r$ apart. This inequality is because every edge we add will either reduce the number of connected components by $1$ or keeps the number of connected components unchanged. Coupling these random variables we have the bound $\mathbb{E}[n-C(n, r)]\leq \mathbb{E}\left[E\left(n,r\right)\right]$. The expected number of edges is much easier to calculate than the expected number of components, with
  \begin{align*}
    \mathbb{E}\left[E\left(n,r_n\right)\right]
    & =\frac{n(n-1)}{2}\mathbb{P}(\lVert x-y \rVert<2r_n\mid x,y \sim X)
    \leq \frac{n(n-1)}{2} Mb_m r_n^m %             \leq \frac{K \pi}{2}r(n)^2
  \end{align*}
  where $M$ is the global upper bound on the probability density function of $X$ and $b_m$ is the volume of the unit ball in $\R^m$.

  Consider the sequence $\left\{r_n=n^{-1/m}g_n\right\}$.
  %such that $n^2r_n^m=\log(\omega_n)$.
  %that is $$r_n=\left(\frac{\log(\omega_n)}{n^2}\right)^{1/m}.$$
  % =\sqrt{2\log (\omega (n))}$ (noting $\lim_{n\to \infty} r(n)= \infty$ as $\lim_{n\to \infty} \omega_n=\infty$) we get
  % $\mathbb{E}\left[C\left(n, \frac{r(n)}{n}\right)\right]\geq n- K \pi\log(\omega_n)$ and thus
  Then
  $$\mathbb{E}[n-C(n, r_n)]\leq \frac{n(n-1)}{2} Mb_m r_n^m \leq \frac{Mb_m}{2} n g_n^m$$
  and so
  $$\mathbb{P}(C(n,r_n)\leq  n-\omega_n)\leq \frac{Mb_m n g_n^m}{2\omega_n}\leq \frac{Mb_m g_n^m}{2\gamma}$$
  which limits to $0$ as $n\to \infty$. This implies that $\lim_{n\to \infty}\mathbb{P}(d_{\omega_n}\geq r_n)=1$.
\end{proof}

\subsection{Admissible bandwidths and rates of convergence}

We can combine Theorem \ref{thm:d_alphan to 0} with Theorem \ref{thm:nd_f(n)}  to show that $d_{\gamma(n-1)}$ is an admissible bandwidth (Definition \ref{def:admissible_bandwidth}).

\begin{theorem}\label{thm:consistency}
  Let $f:\R^m \to \R$ be a density function that satisfies A1. If there exists $0<\gamma_1< \gamma_2<1$ such that $\omega_n \in (\gamma_1 n, \gamma_2 n)$ for all sufficiently large $n$ then $d_{\omega_n}$ is an admissible bandwidth. %for all $\delta\in (0,1)$.
\end{theorem}

\begin{proof}
  Note that the assumptions of A1 imply that $f$ is Lipschitz and bounded.

  Recall that for $d_{\omega_n}$ to be an admissible bandwidth we require $d_{\omega_n} \to 0$ and $n(d_{\omega_n})^{m/2} \to \infty$ as $n\to \infty$ with probability $1$.
  %From Theorem \ref{thm:d_alphan to 0}  $d_{\omega_n} \to 0$ as $n\to \infty$ with probability $1$. As $\delta>0$ is fixed, this also holds for $(d_{\omega_n})^\delta$.
  As $n^{-1/m} \to 0$ as $n\to \infty$ we can conclude from Theorem \ref{thm:d_alphan to 0} that $d_{\omega_n}\to 0$ as $n\to\infty$ (with probability $1$).

  From Theorem \ref{thm:nd_f(n)}, with $g_n=(\log n)^{-1/m}$, we get $\lim_{n\to \infty}\mathbb{P}\left(d_{\omega_n} \geq (n\log n)^{-1/m}\right)=1$ and thus
  $$ \lim_{n\to \infty}\mathbb{P}\left(nd_{\omega_n}^{m/2} \geq n (n\log n)^{-1/2}\right)=1.$$
  Since $n (n\log n)^{-1/2}\to \infty$ as $n\to \infty$ we have in turn  $n(d_{\omega_n})^{m/2}\to \infty$ as $n\to \infty$ with probability $1$.

\end{proof}

Thus, setting the bandwidth to be $\bm{H} = d_{\gamma(n-1)} \bm{I}_m$, for $0<\gamma<1$, leads to a consistent kernel density estimator with MISE \eqref{eq:mise} converging to zero as $n\to \infty$.

Ultimately we are not really interested in optimizing the MISE, as we care more about the tails of the distribution where $f$ is relatively small. However, proving consistency and bounds on mean square error are still useful as a theoretical justification of bandwidth selection.

An asymptotic approximation of the MISE is given by \citep{chacon2018multivariate}
\begin{equation}\label{eq:amse}
  \text{AMISE}(\hat{f}(\cdot, \bm{H})) = n^{-1}|\bm{H}|^{-1/2} R(K) + \frac{1}{4} \mu_2(K)^2 \int \left\{\text{tr}(\bm{H} D^2 f(\bm{u}))\right\}^2 d\bm{u},
\end{equation}
where $R(K) = \int K(\bm{u})^2 d\bm{u}$ and $D^2 f$ is the Hessian of $f$.

If $\bm{H} = h_n \bm{I}_m$ for some $h_n > 0$, then we can write
\begin{equation}\label{eq:amse_iso}
  \text{AMISE}(\hat{f}(\cdot, h_n \bm{I}_m)) = n^{-1} h_n^{-m/2} R(K) + \frac{1}{4} \mu_2(K)^2 h_n^2 \int \left\{\text{tr}(D^2 f(\bm{u}))\right\}^2 d\bm{u}.
\end{equation}

The minimal MISE rate is achieved when $h_n$ is of order $n^{-2/(m+4)}$, and the corresponding MISE is of order $n^{-4/(m+4)}$.

From the upper and lower bounds in Theorem  \ref{thm:d_alphan to 0} and Theorem \ref{thm:nd_f(n)}, asymptotically we expect, with high probability, that $d_{\gamma(n-1)}$ will have a rate close to $n^{-1/m}$. A rate of $n^{-1/m}$ would correspond to AMISE of order $O(n^{-1/2}+n^{-2/m})$.

%We can also use the upper and lower bounds to infer an upper bound on the rate of convergence of the kernel regression estimator using the bandwidth $d_{\omega}$ for suitable sequences $\omega_n$. We can prove that we almost get the classical theoretically optimal bandwidth choice for kernel regression in Euclidean space.

%\begin{corollary}\label{cor:convergence_rate}
%  Let $f\in C^{p}(\R^n, \R)$ be a probability density function for $X$ that satisfies A1. Choose any function $\hat{g}(n)$ with $\hat{g}(n)\to \infty$ as $n\to \infty$.

%  If there exists $0<\gamma_1< \gamma_2<1$ such that $\omega_n \in (\gamma_1 n, \gamma_2 n)$, for all sufficiently large $n$, then the symmetric kernel regression estimator with bandwidth $d_{\omega_n}^{\frac{m}{m+2p}}$ has an asymptotic convergence rate bounded above by $n^{\frac{-2}{m+2p}}\hat{g}(n)$.
%\end{corollary}

%\begin{proof}
%From
%$$\text{AMISE}(\hat{f}(\cdot, h_n \bm{I}_m)) = C_1 n^{-1} h_n^{-m/2} + C_2 h_n^2 $$
%\end{proof}

  There is evidence that $d_*$ (the choice made in the original lookout algorithm) is an admissible bandwidth for samples drawn from a Gaussian distribution, but that it is not an admissible bandwidth when drawn from distributions with fat tails. Given  Theorem \ref{thm:d_n-1 Gaussian} this seems likely but technically challenging. It is not the focus of the new version of lookout, and thus we will leave this proof as an open problem. Part of the challenge is we have no formula $\omega_n$ for where the largest gap in the death values is likely to occur.

\section{Lookout algorithm v2}\label{sec:modified}

In light of the preceding theoretical considerations, we propose three modifications to the lookout algorithm, described in this section.

\subsection{Standardization}

The \emph{orthogonalized Gnanadesikan-Kettenring estimator} provides an iterative robust estimate of $\bm{\Sigma}$, the covariance matrix of $\{\bm{y}_1,\dots,\bm{y}_n\}$ \citep{Maronna2002-qn}. Let
\begin{equation}\label{eq:ogk}
  \hat{\bm{\Sigma}}^{-1} = \bm{U}' \bm{U}
\end{equation}
denote the Cholesky decomposition of its inverse. Then, the data can be scaled by $\bm{U}$, giving $\bm{z}_i = \bm{U}(\bm{y}_i - \bm{m})$, where $\bm{m}$ is the vector of medians of $\{\bm{y}_1,\dots,\bm{y}_n\}$. The covariance matrix of the scaled data $\{\bm{z}_1,\dots,\bm{z}_n\}$ is (approximately) the identity matrix. This can be thought of as both scaling and rotating the data; it is the multivariate equivalent of computing robust z-scores.

We prefer this to the min-max scaling used in the original lookout algorithm, as it is robust to outliers and removes the correlations between variables, leading to more efficient kernel density estimates.

\subsection{Bandwidth matrix selection}

Let $d_1,\dots,d_{n-1}$ denote the finite (ordered) Rips death diameters, computed from the scaled and rotated data $\bm{z}_1,\dots,\bm{z}_n$, and let $d_{\gamma(n-1)}$ denote the $\gamma$-quantile computed from $\{d_1,\dots,d_{n-1}\}$ \citep{HF96}, for a large value of $\gamma \in (0,1)$.

A kernel density estimator is applied to the scaled data $\{\bm{z}_1,\dots,\bm{z}_n\}$ using the bandwidth matrix $\bm{H} = d_{\gamma(n-1)}\bm{I}_m$, where $m$ is the dimension of the data. Under reasonable assumptions, the density estimator is consistent under Theorem \ref{thm:consistency}. When the proportion of anomalies in the data is small, we would expect this estimate to be close to the density of the non-anomalous observations.

Further, the kernels are now aligned with the data, as both have covariance matrices that are scalar multiples of the identity matrix. While alignment is not necessary for consistency, it does improve the efficiency of the kernel density estimates \citep{chacon2018multivariate}.

An equivalent density estimate of the original data $\{\bm{y}_1,\dots,\bm{y}_n\}$ could be obtained using $\bm{H} = d_{\gamma(n-1)}\hat{\bm{\Sigma}}$, although for our purposes, working with the scaled data is sufficient.

\subsection{Bounds for kernel density estimates}

In the following lemma we show that the kernel density estimates at the observations, and their negative logarithms, are bounded by constants that depend only on the dimension $m$ and the sample size $n$.

\begin{lemma}\label{lemma:boundsforkde}
  Let $\{\bm{y}_1, \dots, \bm{y}_n\}$ be $n$ data points in $\R^m$, with kernel density estimate given by
  $$
  \hat{f}(\bm{y}) = \frac{1}{n} \sum_{i = 1}^n |\bm{H}|^{-1/2}K\left(\bm{H}^{-1/2}(\bm{y} - \bm{y}_i)\right),
  $$
  where $K(\bm{u})$ is a multivariate kernel function that is bounded below by 0 and above by $K(\bm{0}) = K_0$, and $\bm{H}$ is an $m\times m$ symmetric positive-definite bandwidth matrix. Then for all $j \in \{1, \dots, n\}$, the kernel density estimates at the observations, $\hat{f}(\bm{y}_j)$, and their negative logarithms, $-\log \hat{f}(\bm{y}_j)$, are bounded by constants that depend only on $m$ and $n$.
\end{lemma}
\begin{proof}
  $$
  \frac{K_0|\bm{H}|^{-1/2}}{n} \leq\hat{f}(\bm{y}_j) =
  \frac{1}{n} \sum_{i = 1}^n |\bm{H}|^{-1/2}K\left(\bm{H}^{-1/2}(\bm{y}_j - \bm{y}_i)\right) \leq K_0|\bm{H}|^{-1/2}\, .
  $$
  The constant $K_0$ may depend on $m$.
\end{proof}

For example, if we use the multivariate Gaussian kernel $K(\bm{u}) = \frac{1}{(2 \pi)^{m/2}} \exp \left( -\frac12 \|\bm{u}\| \right)$, then $K_0 = \frac{1}{(2 \pi)^{m/2}}$. If we use the multivariate Epanechnikov kernel $K(\bm{u}) = \frac{m+2}{2b_m} \left( 1 - \|\bm{u}\| \right)_{+}$, where $b_m$ is the volume of the unit sphere in $\R^m$, and $u_+ = \max(0, u)$, then $K_0 = (m+2)/(2b_m)$.

Under Theorem \ref{thm:FisherTippett}, the limiting distribution of the extremes of $-\log\hat{f}(\bm{y}_j)$ is a Weibull distribution, corresponding to a Generalized Extreme Value distribution with $\xi<0$. Thus, when we fit a Generalized Pareto Distribution (GPD) to the largest $1-\beta$ of $\{-\log\hat{f}(\bm{y}_j)\}_{j=1}^n$, we can constrain the shape parameter to be non-positive, which improves the stability of the GPD fitting.

\subsection{Modified lookout algorithm}

We are now in a position to describe the modified lookout algorithm.

\begin{algorithm}
  \caption{Modified lookout algorithm}\label{alg:modifiedlookout}
  \begin{algorithmic}[1]
    \Require Data $\{\bm{y}_1, \dots, \bm{y}_n\}$, proportion $\beta \in (0,1)$, threshold $\alpha \in (0,1)$, and quantile level $\gamma \in (0,1)$.
    \State Let $\bm{Y}$ be the $m\times n$ data matrix with rows $\bm{y}_1, \dots, \bm{y}_n$.
    \State Optional scaling: Calculate the eigendecomposition \eqref{eq:ogk}, and rotate and scale the data using $\bm{U}$, giving $\bm{Z} = \bm{U}\bm{Y}$, with scaled observations $\bm{z}_1, \dots, \bm{z}_n$ comprising the rows of $\bm{Z}$.
    \State Compute the persistence homology barcode of the data cloud $\bm{Z}$ (set $\bm{Z} = \bm{Y}$ if no scaling is applied) for dimension zero using the Vietoris-Rips diameter \citep{ghrist2008barcodes}, and obtain the ordered death diameters $\{d_i\}_{i = 1}^n$ for connected components.
    \State Let $d_{\gamma(n-1)} =$ the $\gamma$ sample quantile computed from $\{ d_i \}_{i=1}^n$.
    \State Compute the kernel density estimates at the observations: $f_i = \hat{f}(\bm{y}_i)$, $i=1,\dots,n$, where $\hat{f}$ is given by \eqref{eq:mkde}, $\bm{H} = h_n\bm{I}_m$, and $h_n = d_{\gamma(n-1)}$.
    \State Compute the leave-one-out kde values $f_{-i} = \frac{1}{n-1} \big( n f_i - \bm{H}^{-1/2} K(\bm{0}) \big)$ for $i=1,\dots,n$.
    \State Fit a Generalized Pareto Distribution (GPD) to the largest $1-\beta$ of the surprisals $\{-\log f_i\}_{i=1}^n$, constraining the shape parameter to be non-positive, obtaining estimates of the location, scale and shape parameters $(\mu, \sigma, \xi)$.
    \State Calculate the probability of each observation by applying the estimated GDP to the leave-one-out kde values: $p_i = (1-\beta)P(-\log f_{-i} \mid \hat{\mu}, \hat{\sigma}, \hat{\xi})$.
    \State If $p_i < \alpha$, declare $\bm{y}_i$ anomalous.
  \end{algorithmic}
\end{algorithm}

As with the original lookout algorithm, the values of $p_i$ can also act as anomaly scores, with lower values indicating more anomalous points.

The default values for the parameters are $\alpha = 0.001$, $\beta = 0.90$ and $\gamma = 0.98$. The default values are used in all experiments reported here, unless explicitly stated otherwise. Scaling is also applied by default, unless otherwise stated.

% ---------------------------------------------------------------------------
\section{Comparing new lookout with the older version}\label{sec:comparewitholderversion}
% ---------------------------------------------------------------------------

Next we compare the old and the new versions of lookout. We conduct four experiments that explore different aspects of the updated algorithm, and illustrate the output using a set of showcase examples. All experiments have anomalies and non-anomalies. Experiments 1 and 2 allow the anomalies to move away from the non-anomalous distribution. These two experiments compare the sensitivity of the old and new versions of lookout to identifying anomalies as they become more obvious. Experiments 3 and 4 focus on having anomalies on the boundary of the non-anomalous region. For these two experiments we increase the number of data points and explore the performance of the two algorithm versions for different sample sizes.

% - Show anomaly detection examples using normal distribution and
% gamma distribution, with lookout 2 vs lookout 1

\subsection{Experiment 1: Different anomaly rates for Gamma distribution}

For experiment 1 we consider the $\text{Gamma}(a, r)$ distribution with shape $a$ and rate $r$ and generate data points with different rates $r$. We select $\text{Gamma}(2, 2)$ as the non-anomalous distribution and $\text{Gamma}(2, r)$ as the anomalous distribution where $r \in \{0.1, 0.2, \ldots, 1\}$. Let $Y_i = (Y_{i1}, Y_{i2})$ where $Y_{i1}$ and $Y_{i2}$ are independent with $Y_{ij} \sim \text{Gamma}(2, 2)$ for the non-anomalous distribution and $Y_{ij} \sim \text{Gamma}(2, r)$ for the anomalous distributions. We consider 10 replications of the experiment for each value of $r$. For each replication we generate 500 non-anomalies and 10 anomalies.
\begin{figure}[!ht]
  \centering
  \includegraphics[width=\textwidth]{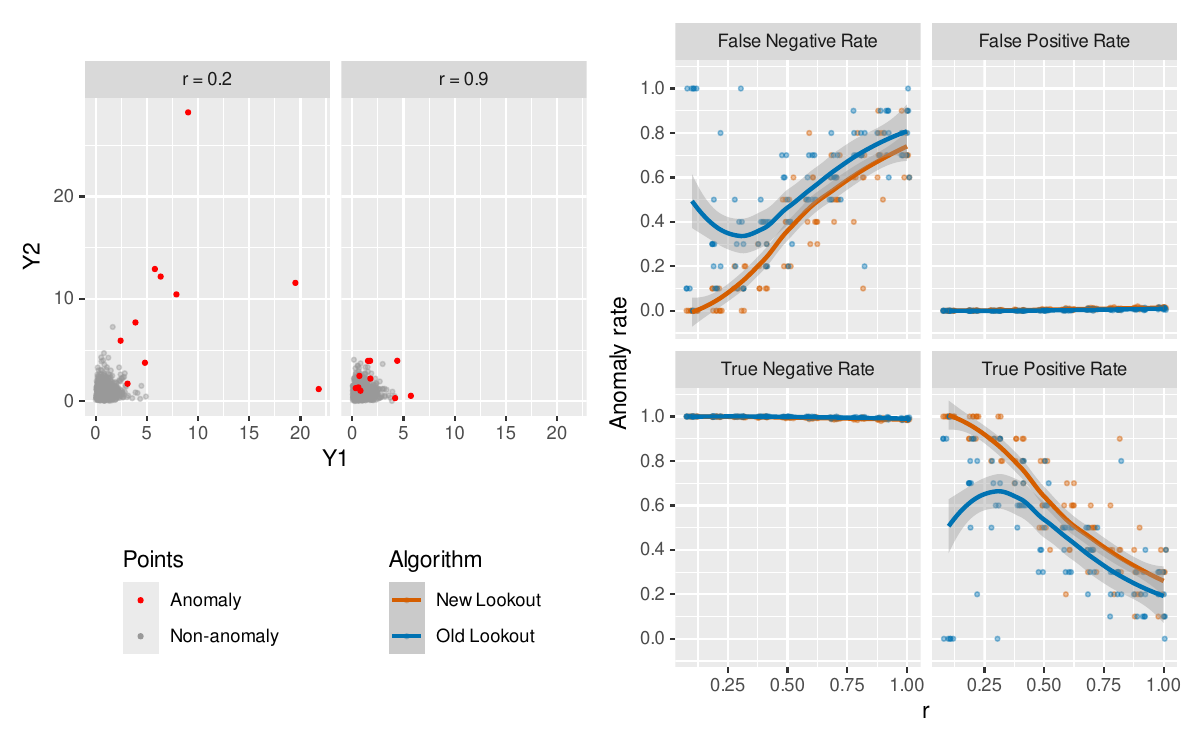}
  \caption{Experiment 1 data and results. Left: data for two iterations of the experiment with $r = 0.2$ and $r = 0.9$. Anomalies with lower values of $r$ are easier to identify. Right: the relative performance of the two algorithms. The points in the right panel are slightly jittered horizontally to avoid overlapping points.}
  \label{fig:exp1GammaRates}
\end{figure}

\cref{fig:exp1GammaRates} shows the data and results of Experiment 1. As the rate $r$ of the anomalous Gamma distribution increases, the anomalies and non-anomalies become more similar. Thus, we expect a decrease in the true positive rate and an increase in the false negative rate as $r$ increases. The anomalies are maximally different from non-anomalies when $r = 0.1$. We see that new lookout performs better on average than the old lookout algorithm across all values of $r$, with the largest differences for the smallest values of $r$.

\subsection{Experiment 2: Different anomaly means for Normal distribution}

For experiment 2 we consider the normal distribution $\mathcal{N}(\mu, 1)$ with mean $\mu$. We let $Y_i = (Y_{i1}, Y_{i2})$ where $Y_{i1}$ and $Y_{i2}$ are independent with $Y_{ij} \sim \mathcal{N}(\mu, 1)$. The non-anomalous distribution has $\mu=0$, while the anomalous distribution has $\mu \in \{ 2.50, 2.75, \dots, 3.75, 4.00\}$. We consider 1000 non-anomalies and 10 anomalies, and we produce 10 replicates for each value of $\mu$.

\cref{fig:exp2NormalRates} shows Experiment 2 data and results. As $\mu$ increases, the anomalies shift away from the non-anomalies. Thus, we expect the true positive rate to increase and the false negative rate to decrease with $\mu$. Here, the two versions of the lookout algorithm are statistically indistinguishable, across all values of $\mu$.

\begin{figure}[!hb]
  \centering
  \includegraphics[width=1\textwidth]{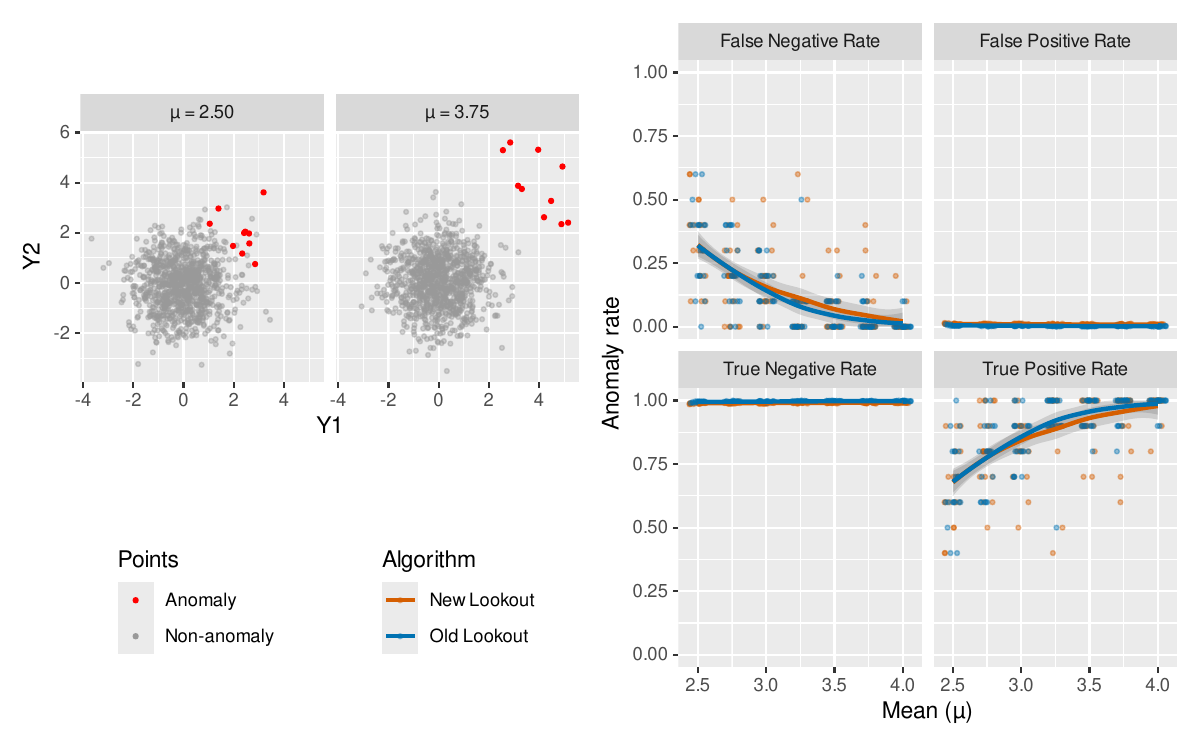}
  \caption{Experiment 2 data and results. Left: data for two iterations of the experiment with $\mu = 2.5$ and $\mu = 3.75$. Anomalies with higher values of $\mu$ are easier to identify. Right: the relative performance of the two algorithms. The points in the right panel are slightly jittered horizontally to avoid overlapping points.}
  \label{fig:exp2NormalRates}
\end{figure}

\subsection[Experiment 3: As N increases - Normal distribution]{Experiment 3: As $N$ increases - Normal distribution}

It is often difficult to distinguish anomalies from legitimate observations because they can appear similar. In such cases, anomalies are not well separated from the non-anomalous points. In this experiment, we explore the efficacy of \textit{lookout} when anomalies are close to regions containing non-anomalous points.

We consider $n$ non-anomalies for $n \in \{1000, 2000, \ldots, 10000\}$ and $0.005 \times n$ anomalies. We let $Y_i = (Y_{i1}, Y_{i2})$ where $Y_{i1}$ and $Y_{i2}$ are independent, with $Y_{ij} \sim \mathcal{N}(\mu, \sigma)$. For the non-anomalies, we set $\mu=0$ and $\sigma=1$. We generate anomalies in a ring just beyond most of the non-anomalous points, by first sampling $\mu_{i1} \sim \mathcal{U}(-\sqrt{2}\times 2.2, \sqrt{2}\times 2.2)$ and letting $\mu_{i2} = \sqrt{2\times2.2^2 - \mu_{i1}^2}$. Then we sample the anomalies $(Y_{i1}, Y_{i2})$ where $Y_{i1} \sim \mathcal{N}(\mu_{i1}, 0.1)$ and $Y_{i2} \sim \mathcal{N}(\mu_{i2}, 0.1)$. We produce 10 replications for each value of $n$.

\cref{fig:exp3NormalIncreasingN} shows the data of a single run and the results of experiment 3. Both new and old lookout give low false positive rates and high true negative rates. However, we see that new lookout has a larger true positive rate and a lower false negative rate compared to old lookout, with statistically significant differences for smaller sample sizes.

\begin{figure}[!hb]
  \centering
  \includegraphics[width=\textwidth]{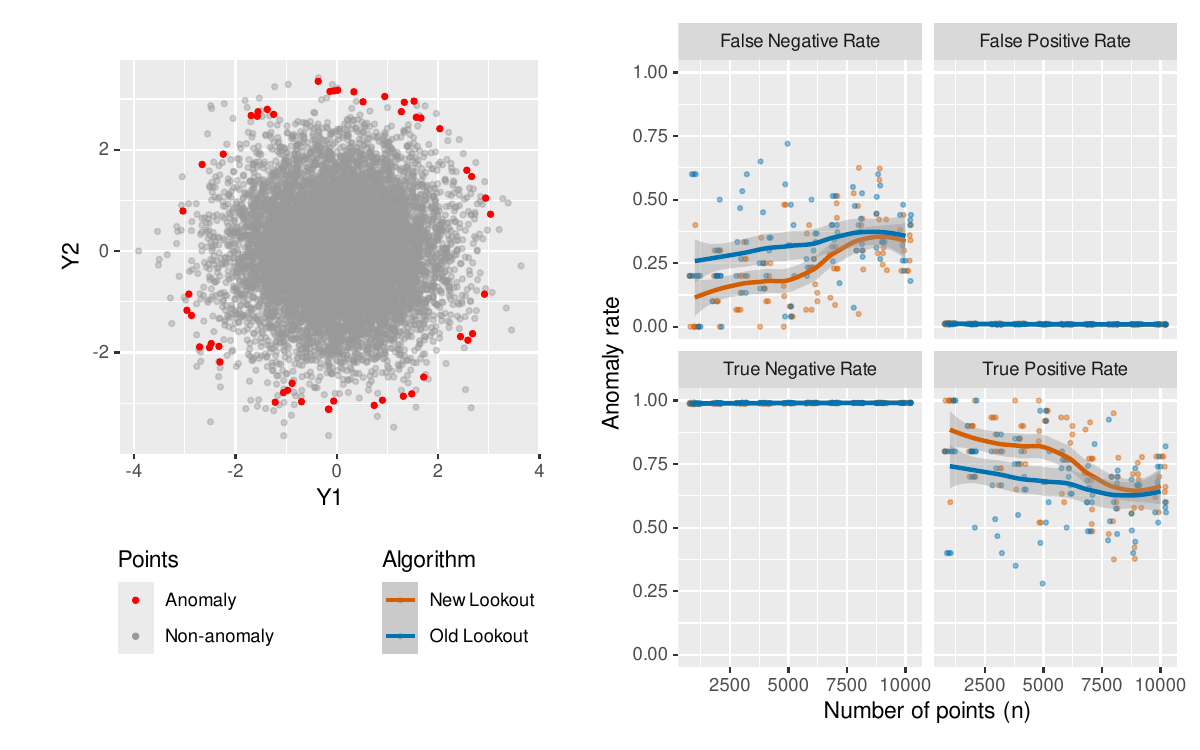}
  \caption{Experiment 3 data and results using Normal distributions. Left: data for $n=10000$. Right: the relative performance of the two algorithms. The points in the right panel are slightly jittered horizontally to avoid overlapping points.}
  \label{fig:exp3NormalIncreasingN}
\end{figure}

\subsection[Experiment 4: As N increases - Gamma distributions]{Experiment 4: As $N$ increases - Gamma distribution}

This experiment is similar to experiment 3, but we set $Y_{ij} \sim \text{Gamma}(2, 2)$ for the non-anomalous points. To obtain anomalous points, we first generate $n$ points $Y_i = (Y_{i1},Y_{i2})$ where $Y_{i1}$ and $Y_{i2}$ are independent, with $Y_ij \sim \text{Gamma}(2.2, 2)$. Then we consider the subset $\{Y_i: \lVert Y_i \rVert > \lVert X\rVert _{0.99} \}$ where $\lVert Y_i \rVert = \sqrt{ Y_{i1}^2 + Y_{i2}^2} $ denotes the distance from the origin to $Y_i$ and $\lVert X\rVert _{0.99}$ denotes the 0.99 quantile of $\left \{\sqrt{ Y_{i1}^2 + Y_{i2}^2} \right \}_{i=1}^n$. From this subset we sample $0.005 \times n$ points and label them as anomalies. Each iteration for a given $n$ is repeated 10 times. \cref{fig:exp4GammaIncreasingN} shows the data from an iteration of the experiment with $n = 10000$ and the results for the new and old version of \textit{lookout}. The two versions of lookout are not signficantly different, other than at the largest sample sizes, where the older version does slightly better.

\begin{figure}[!thb]
  \centering
  \includegraphics[width=\textwidth]{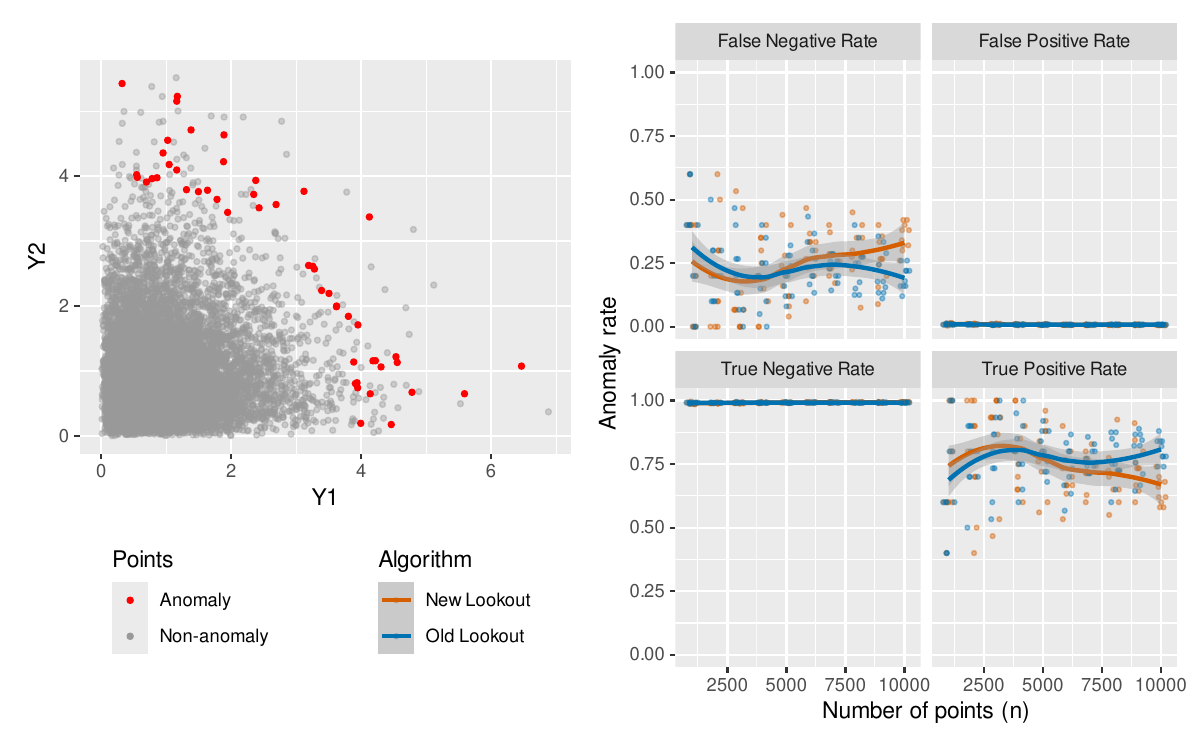}
  \caption{Experiment 4 data and results using Gamma distributions. Left: data for $n=10000$. Right: the relative performance of the two algorithms. The points in the right panel are slightly jittered horizontally to avoid overlapping points.}
  \label{fig:exp4GammaIncreasingN}
\end{figure}

\subsection{Showcase examples}\label{sec:realworld}

\begin{figure}[!p]
  \centering
  \includegraphics[width=0.84\linewidth]{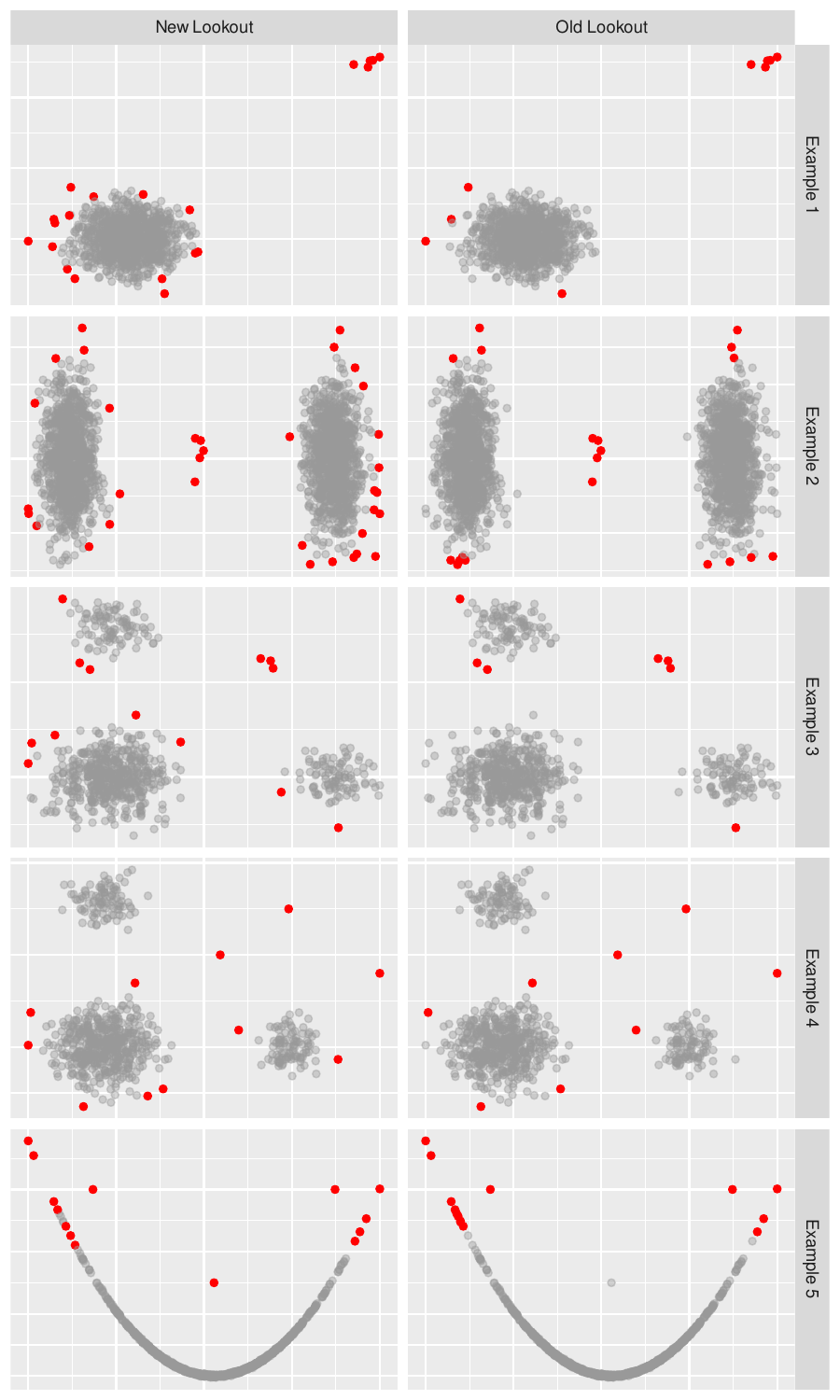}
  \caption{Comparing new lookout with the older version on some showcase examples}
  \label{fig:showcaseExamples}
\end{figure}

\cref{fig:showcaseExamples} shows some examples of new and old versions of lookout with identified anomalies in red. The purpose of this exercise is to check if the new version agrees with the old version on obvious anomalies. The graphs on the first row have anomalies at the top right-hand corner, which are identified by both algorithms; both algorithms find some false anomalies in the lower cloud of points, with the new version finding more. The second row depicts a bi-modal distribution with anomalies placed in the middle. Both versions of lookout identify the central anomalies, along with some false anomalies, with the new version finding more false anomalies. The third row has 3 clusters of normally distributed points with anomalies placed away from the three clusters. Here, the older version of the algorithm finds four false positives compared to ten with the new algorithm. The same setting is repeated in the fourth row with the anomalies spread out, and the two algorithms are similar. The last row has non-anomalous points in a parabola and the anomalous points outside of it. This is a tricky example as the points are on a manifold, and we are not considering their geometry. In this example, only the new algorithm finds all three anomalies lying off the manifold, while both algorithms producing several false positives on the manifold. Overall, on these examples, the new algorithm is a more sensitive to identifying anomalies, but also produces a few more false positives. The false anomalies identified by the new algorithm tend to be in regions of the data space with a low density of points, which is consistent with our definition of an anomaly.

\subsection{Real-world Examples}

We check the new and old versions of lookout on two sets of real data. Neither of these examples have labeled anomalies and so we cannot compute accuracy measures. However, as both datasets have a small number of variables, we can visualize the output. The datasets are
\begin{enumerate}
  \item Old faithful dataset\footnote{Data downloaded from \url{https://geysertimes.org}.}: eruptions of the Old Faithful Geyser in Yellowstone National Park, Wyoming, USA, from 14 January 2017 to 29 December 2023. For each recorded eruption the duration is given in seconds along with the waiting time to the following eruption. The records are incomplete, especially during winter.
  \item Wine reviews dataset\footnote{Data downloaded from \url{https://www.kaggle.com/datasets/zynicide/wine-reviews/data}.}: different attributes of wines from 44 countries, obtained from the \emph{Wine Enthusiast Magazine} during the week of 15 June 2017. We focus on the numerical points score given by the \emph{Wine Enthusiast} reviewer on a scale of 0--100 and the corresponding price of each bottle of wine.
\end{enumerate}

\begin{figure}[!htb]
  \centering
  \includegraphics[width=\linewidth]{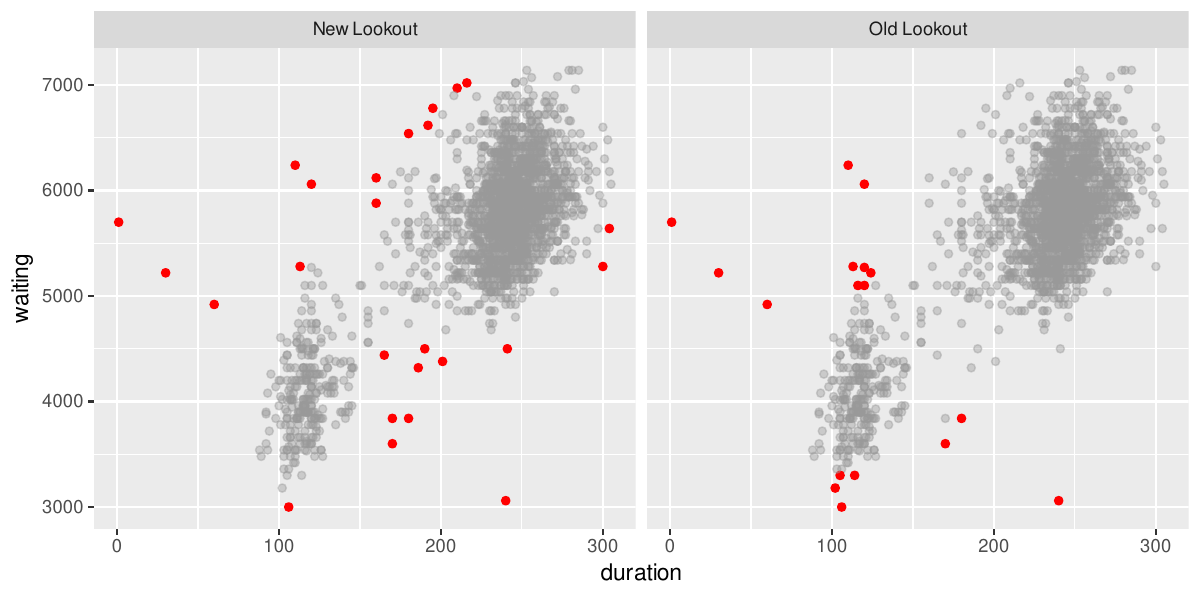}
  \caption{Anomalies in the Old Faithful dataset}
  \label{fig:oldfaithful}
\end{figure}

\begin{figure}[!htb]
  \centering
  \includegraphics[width=\linewidth]{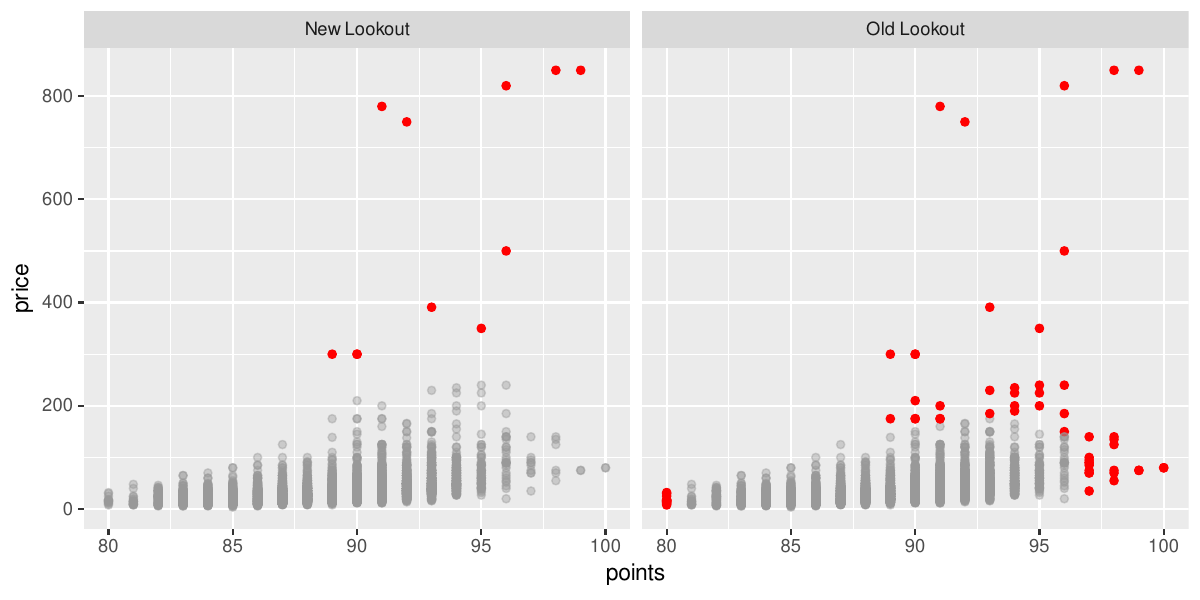}
  \caption{Anomalies in the Wine Reviews dataset. Left with new version of lookout and the right with the old version of lookout.}
  \label{fig:winereviews}
\end{figure}

\cref{fig:oldfaithful} shows the output of new and old versions of lookout when applied to the \textit{Old Faithful} dataset, with the identified anomalies depicted in red. New lookout identifies more ``anomalous'' points in regions of low density, while the old version tends to identify some points that are close to the main cloud of points as anomalous, and misses some points in the low density regions. In this example, the new lookout algorithm appears to be an improvement over the old.

\cref{fig:winereviews} shows the output of lookout on the \textit{Wine Reviews} dataset. For this dataset we see that the new version identifies far fewer anomalies compared the old version, especially at the extremes of the points score, and around the main cloud of points. In this example the new lookout algorithm is a substantial improvement over the old.

% ---------------------------------------------------------------------------
\section{Evaluating lookout with other anomaly detection methods}\label{sec:compareWithOtherMethods}
% ---------------------------------------------------------------------------

In this section we compare the new and old versions of lookout with two Extreme Value Theory (EVT) based methods and two KDE-based methods. The EVT-based methods are HDoutliers \citep{wilkinson2017visualizing} and \textit{stray} \citep{talagala2019anomaly}, both with their $\alpha$ parameter set to zero (which is the default for stray). The KDE-based methods are KDEOS \citep{Schubert2014} and RDOS \citep{Tang2017}, both run with their default settings. The same anomaly detection methods were used for comparison in the original version of lookout \citep{lookout}.

Both HDoutliers and stray identify anomalies, i.e., they give a binary label 1 to anomalies and 0 to non-anomalies. Therefore, to compare lookout with stray and HDoutliers we have used the metrics \textit{Fmeasure} and \textit{Gmean}. Suppose TP denotes the true positives (actual = 1, predicted = 1), FP denotes the false positives (actual = 0, predicted = 1), TN denotes the true negatives (actual = 0, predicted = 0) and FN denotes the false negatives (actual = 1, predicted = 0). Then these metrics are defined as follows:
\[
  \text{Fmeasure} = \frac{2(\text{Precision}\times\text{Recall})}{\text{Precision} + \text{Recall}}
  \quad \text{where} \quad
  \text{Precision} = \frac{\text{TP}}{\text{TP} + \text{FP}}
  \quad \text{and} \quad
  \text{Recall} = \frac{\text{TP}}{\text{TP} + \text{FN}} \, ,
\]
\[
  \text{Gmean} = \sqrt{\text{Sensitivity} \times \text{Specificity} }
  \quad \text{where} \quad \text{Sensitivity} = \frac{\text{TP}}{\text{TP} + \text{FN}}
  \quad \text{and} \quad \text{Specificity} = \frac{\text{TN}}{\text{TN} + \text{FP}} \, .
\]
Note that recall and sensitivity define the same quantity. Sensitivity is commonly used in epidemiological settings while recall is used in computer science.

The KDE-based methods KDEOS and RDOS give anomaly scores instead of identifying points as anomalous or not. A standard way to compare anomaly scores is by using the area under the Receiver Operator Curve (ROC). For a given set of anomaly scores the ROC plots 1 - Specificity on the $x$-axis for different cut-off points and Sensitivity on the $y$-axis. A method $A$ is better than method $B$ if the area under the ROC is greater for method $A$ compared to $B$. As lookout gives anomaly scores as well as binary labels, we compare lookout with KDEOS and RDOS using the area under the ROC, which we denote by AUC.

\subsection{Experiment 5: Normal distributions moving out in each iteration}

In this experiment we consider 405 points in $\R^6$, of which 5 points are anomalies. The experiment is run for 10 iterations with anomalies becoming more obvious in each iteration. The non-anomalies are distributed as $\mathcal{N}(0,1)$ in all dimensions. The anomalies differ the non-anomalies only in the first dimension. The anomalies are distributed as $\mathcal{N}\left(2 + (i-1)\times0.5, \, 0.2 \right)$ for iterations $i \in \{1, \ldots, 10\}$ in the first dimension. In the other dimensions the anomalies are distributed as $\mathcal{N}(0,1)$. We run the experiment for 10 iterations and in each iteration the anomalies move further away from the non-anomalous points. Each iteration is repeated 10 times to account for randomness.

\cref{fig:comparisonExp5} shows the data and the results of this experiment, with anomalies shown in red. The top left plot shows the data on the $Y_1Y_2$ plane for the 3rd iteration and the middle two plots show the data for the 9th iteration. We see that the anomalies have separated from the other points in the $Y_1$ direction, but mixed with the non-anomalies on the $Y_3Y_4$ plane.

\begin{figure}[!ht]
  \centering
  \includegraphics[width=\linewidth]{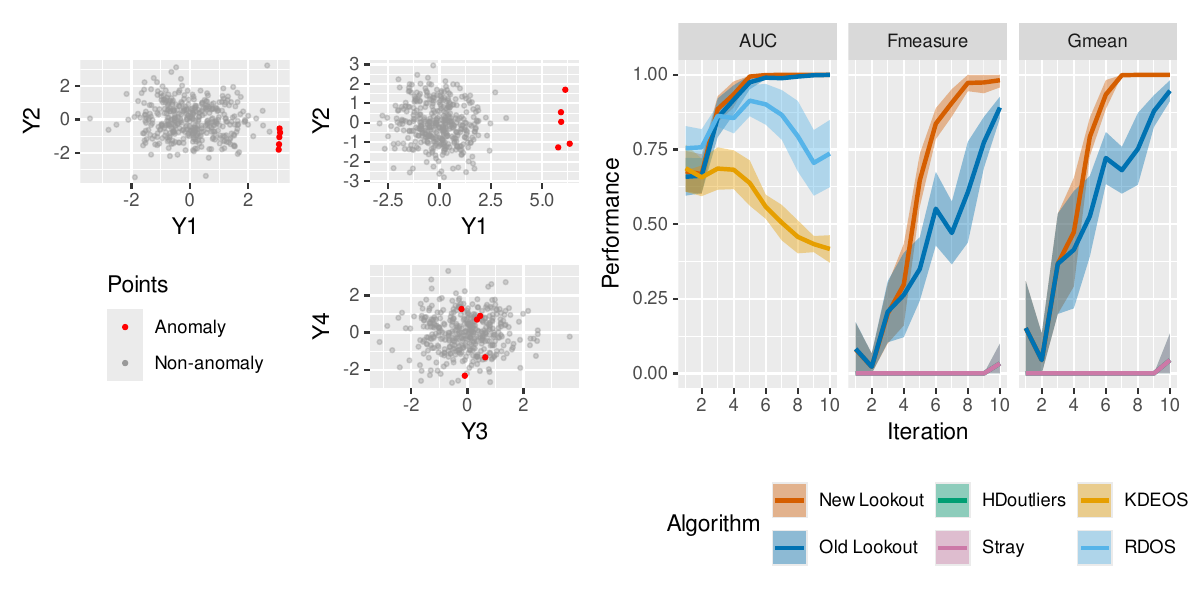}
  \caption{Experiment 5 results. Top left: Data from the 3rd iteration showing the $Y_1Y_2$ plane with anomalies in red. Middle left: Data from the 9th iteration showing both the $Y_1Y_2$ plane and $Y_3Y_4$ plane, with anomalies moving further away in the $Y_1$ direction, resulting in the anomalies and non-anomalies being mixed on the $Y_3Y_4$ plane. Right: AUC, Fmeasure and Gmean of old and new versions of lookout, HDoutliers, KDEOS, RDOS and stray.}
  \label{fig:comparisonExp5}
\end{figure}

The plot on the right gives the AUC when comparing lookout with KDEOS and RDOS and Fmeasure and Gmean for comparison with HDoutliers and stray. As discussed the reason for different metrics is because KDEOS and RDOS give anomaly scores, while stray and HDoutiers identify anomalies. We see that new lookout performs better than the old lookout in terms of Gmean and Fmeasure for higher iterations. In terms of AUC both versions of lookout give the same performance.

\subsection[Experiment 6: Annulus in R4]{Experiment 6: Annulus in $\R^4$}

For this experiment we consider 805 points in $\R^3$ of which 5 points are anomalies. \cref{fig:comparisonExp6} shows the data lying in an annulus on the $Y_1Y_2$ plane with $Y_3 \sim \mathcal{U}(0,1)$. As before, the experiment is run for 10 iterations with anomalies moving towards the centre of the hole in each iteration. Plots from iteration 3 and 9 are shown on the left in \cref{fig:comparisonExp6}. The results of the different methods are shown on the right.

\begin{figure}[!ht]
  \centering
  \includegraphics[width=\linewidth]{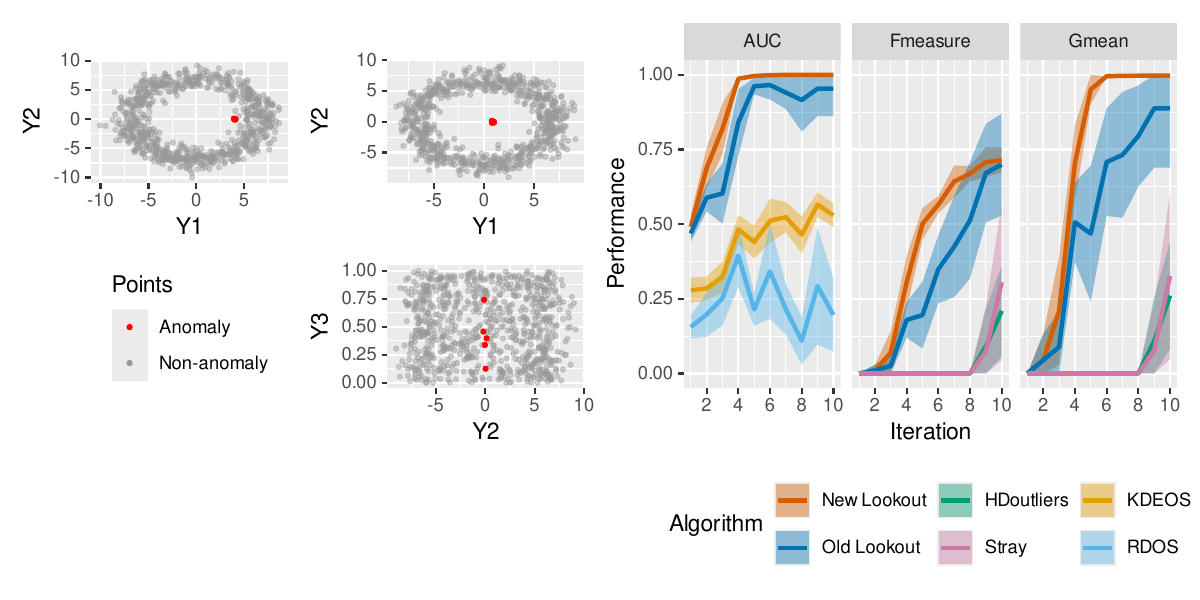}
  \caption{Experiment 6 results. Top left: Data from the 3rd iteration showing the $Y_1Y_2$ plane with anomalies in red. Middle left: Data from the 9th iteration showing both the $Y_1Y_2$ plane and $Y_2Y_3$ plane, with anomalies moving into the annulus. Right: AUC, Fmeasure and Gmean of old and new versions of lookout, HDoutliers, KDEOS, RDOS and stray.}
  \label{fig:comparisonExp6}
\end{figure}

\subsection[Experiment 7: Unit cube in R20]{Experiment 7: Unit cube in $\R^{20}$}

This experiment considers 500 points in $\R^{20}$ of which one is an anomaly. At the start of the experiment all 500 points are uniformly distributed with $y_i \sim \mathcal{U}(0,1)$ in each dimension. For this experiment we consider 20 iterations and in each iteration the 500th data point becomes more anomalous. In the first iteration the 500th point is updated with the first coordinate $Y_1 = 0.9$. In the second iteration it is further updated with $Y_2 = 0.9$, resulting in $(0.9, 0.9, Y_3, \dots, Y_{20})$. In the $i$th iteration the 500th point is updated with $Y_i = 0.9$. Each iteration is repeated 10 times to account for randomness.

The anomalous point is not anomalous in any single dimension or on any $Y_iY_j$ plane. In fact even in the space of the first two principal components, the anomaly is not distinguishable. We use \textit{dobin} \citep{dobinPaper}, a dimension reduction method especially developed for anomalies to visualize this space. \cref{fig:comparisonExp7} shows the data and the results of this experiment. The plots on the left show the data on the $D_1D_2$ plane for iterations 5, 12 and 20, where $D_1$ and $D_2$ are the first two axes of dobin projection. We see that in the 20th iteration the anomaly can be easily identified on the $D_1D_2$ plane. The plot on the right shows the results of applying the various algorithms to this problem. The lookout algorithms are applied without scaling in this experiment.

\begin{figure}[!ht]
  \centering
  \includegraphics[width=\linewidth]{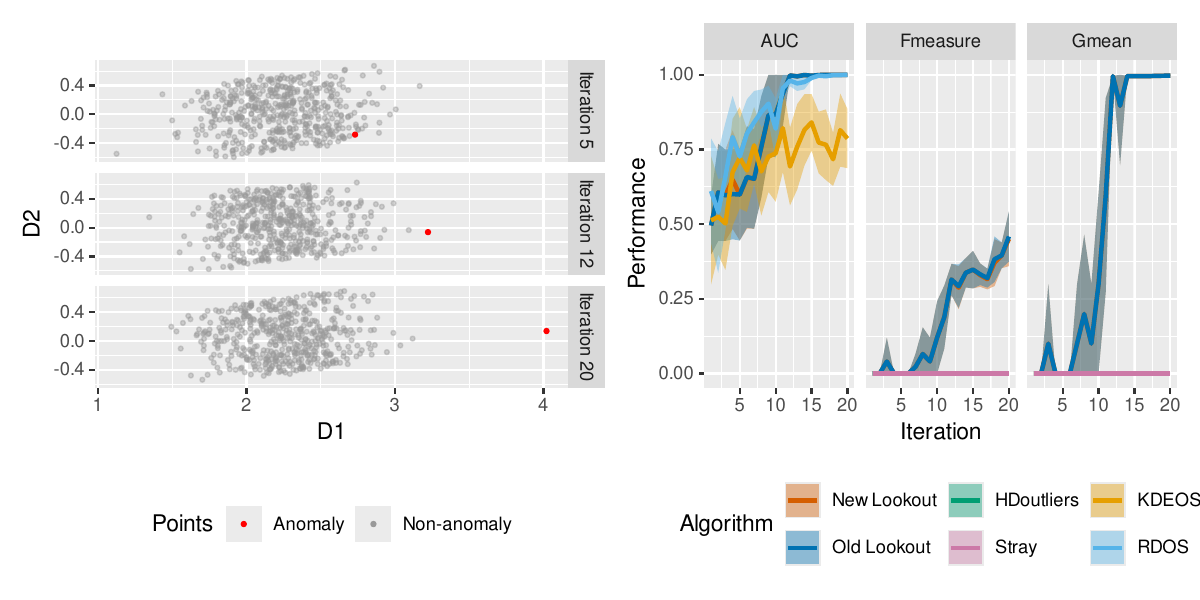}
  \caption{Experiment 7 results. Left: Data from the 5th, 12th and 20th iterations, showing the $D_1D_2$ plane from the dobin projection. Anomalies are shown in red. Right: AUC, Fmeasure and Gmean of old and new versions of lookout, HDoutliers, KDEOS, RDOS and stray.}
  \label{fig:comparisonExp7}
\end{figure}

% ---------------------------------------------------------------------------
\section{Discussion and conclusions }\label{sec:conclusion}
% ---------------------------------------------------------------------------

We have proposed an improved version of the lookout algorithm based on theoretical underpinnings in persistent homology. The updated algorithm uses an upper quantile of the Rips death diameters rather than the maximum successive difference for bandwidth selection, which we prove yields a consistent kernel density estimator under broad conditions. Additionally, we replace min-max scaling with robust multivariate scaling using the orthogonalized Gnanadesikan-Kettenring estimator, improving robustness to outliers and accounting for correlations between variables. Finally, we constrain the shape parameter of the Generalized Pareto Distribution to be non-positive when fitting to the negative log kernel density estimates, leveraging known bounds on these values to enhance stability.

Our empirical evaluation demonstrates that the updated lookout performs better than the original version across diverse examples, while maintaining competitive performance against other anomaly detection methods. The algorithm provides probabilistic anomaly scores with theoretical grounding in extreme value theory, enabling principled threshold selection.

The key strengths of the updated lookout algorithm include: (1) theoretical consistency guarantees for the kernel density estimator under mild conditions; (2) robustness to outliers in both the scaling and bandwidth selection steps; (3) principled probabilistic interpretation through extreme value theory; and (4) automatic adaptation to the local topological structure of the data through persistent homology.

The algorithm is versatile and suited to any continuous multivariate data, and can handle moderate dimensionality effectively (the experiments show good results up to about 20 dimensions) effectively. It can also be applied to any problems that can be transformed into a continuous multivariate setting, such as image anomaly detection via feature extraction, or functional data anomaly detection via functional principal components. We also assume that the data are independent and identically distributed, making it unsuitable for temporal or spatial data without pre-processing (e.g., applying it to the residuals of a fitted model).

However, several computational limitations remain. The computational complexity is dominated by the computation of dimension zero Vietoris-Rips persistent homology, which is equivalent to single link clustering, requiring $O(n\log n)$ time \citep{Gower1969}. This can be prohibitive for very large datasets, although approximate methods or subsampling could mitigate this. The method may also struggle with very high-dimensional data due to the curse of dimensionality affecting kernel density estimation. Using dimension reduction methods, such as principal component analysis or dobin \citep{dobin}, before applying lookout could help in such cases.

\backmatter

\bmhead{Supplementary information}

The updated algorithm is implemented in the \texttt{lookout} R package \citep{Rlookout}, available on CRAN at \url{https://CRAN.R-project.org/package=lookout}.

Files to reproduce all numerical examples are available at \url{https://github.com/robjhyndman/Look-Out-2}

\bmhead{Acknowledgements}

Rob Hyndman and Sevvandi Kandanaarachchi are part of the Australian Research Council Industrial Transformation Training Centre in Optimisation Technologies, Integrated Methodologies, and Applications (OPTIMA), Project ID IC200100009. KT is supported by an Australian Research Council Future Fellowship, FT250100951.

\bmhead{Conflict of interest}
Katharine Turner is a member of the APCT Editorial Board. She was not involved in the peer review or handling of this manuscript. No other conflicts of interest are declared.

\bibliography{References}

\end{document}